\documentclass[10pt]{article}

\usepackage{graphicx}
\usepackage{cite}
\usepackage{hyperref}
\usepackage{amsmath}
\usepackage{amssymb}
\usepackage{amsfonts}
\usepackage{comment}
\usepackage[nofillcomment,noend]{algorithm2e}
\usepackage{appendix}
\usepackage{subfloat}
\usepackage{subfig}
\usepackage{multirow}
\usepackage{url}
\usepackage{fullpage}

\usepackage{mathptmx}
\usepackage{amsthm} 

\input{commands-tam.sty}
\input{numbering-tam.sty}

\newif\ifabstract
\newif\iffull

\abstracttrue

\ifabstract
	\fullfalse
\else
	\fulltrue
\fi

\newcounter{section-preserve}
\newcounter{theorem-preserve}
\newcounter{lemma-preserve}
\newcommand{\blank}[1]{}
\newtoks\magicAppendix
\magicAppendix={}
\newtoks\magictoks
\newif\iflater
\laterfalse
\ifabstract
\long\def\later#1{\magictoks={#1}%
  \edef\magictodo{\noexpand\magicAppendix={\the\magicAppendix \par
    \the\magictoks%
  }}
  \magictodo}
\long\def\both#1{\magictoks={#1}%
  \edef\magictodo{\noexpand\magicAppendix={\the\magicAppendix \par
    \noexpand\setcounter{theorem-preserve}{\noexpand\arabic{theorem}}%
    \noexpand\setcounter{theorem}{\arabic{theorem}}%
    \noexpand\setcounter{lemma-preserve}{\noexpand\arabic{lemma}}%
    \noexpand\setcounter{lemma}{\arabic{lemma}}%
    \noexpand\setcounter{section-preserve}{\noexpand\arabic{section}}%
    \noexpand\setcounter{section}{\arabic{section}}%
	\noexpand\let\noexpand\oldsection=\noexpand\thesection
	\noexpand\def\noexpand\thesection{\thesection}
	\noexpand\let\noexpand\oldlabel=\noexpand\label
	\noexpand\let\noexpand\label=\noexpand\blank
    \the\magictoks%
    \noexpand\setcounter{theorem}{\noexpand\arabic{theorem-preserve}}%
    \noexpand\setcounter{section}{\noexpand\arabic{section-preserve}}%
	\noexpand\let\noexpand\thesection=\noexpand\oldsection
	\noexpand\let\noexpand\label=\noexpand\oldlabel
  }}
  \magictodo
  \the\magictoks}
\else
\long\def\later#1{#1}
\long\def\both#1{#1}
\fi
\long\def\magicappendix{
	\latertrue%
	\the\magicAppendix%
}

\begin{document}

\title{\textbf{Scaled pier fractals do not strictly self-assemble}\footnote{This work is an extension of the conference paper titled, ``Scaled tree fractals do not strictly self-assemble'' \cite{ScaledTreeFractals}.}}

\author{%
David Furcy\thanks{Department of Computer Science
  University of Wisconsin--Oshkosh, Oshkosh, WI 54901, USA.
  \url{furcyd@uwosh.edu}}
\and
Scott M. Summers\thanks{Department of Computer Science
  University of Wisconsin--Oshkosh, Oshkosh, WI 54901, USA. \url{summerss@uwosh.edu}}
}

\date{}
\maketitle

\begin{abstract}
A \emph{pier fractal} is a discrete self-similar fractal whose
generator contains at least one \emph{pier}, that is, a member of the
generator with exactly one adjacent point. Tree fractals and
pinch-point fractals are special cases of pier fractals. In this
paper, we study \emph{scaled pier fractals}, where a \emph{scaled
  fractal} is the shape obtained by replacing each point in the
original fractal by a $c \times c$ block of points, for some $c \in
\Z^+$.  We prove that no scaled discrete self-similar pier fractal
strictly self-assembles, at any temperature, in Winfree's abstract
Tile Assembly Model.
\end{abstract} 

\section{Introduction}
The stunning, often mysterious complexities of the natural world, from nanoscale crystalline structures to unthinkably massive galaxies, all arise from the same elemental process known as \emph{self-assembly}. In the absence of a mathematically rigorous definition, self-assembly is colloquially thought of as the process through which simple, unorganized components spontaneously combine, according to local interaction rules, to form some kind of organized final structure. A major objective of nanotechnology is to harness the power of self-assembly, perhaps for the purpose of engineering atomically precise medical, digital and mechanical components at the nanoscale. One strategy for doing so, developed by Nadrian Seeman, is \emph{DNA tile self-assembly} \cite{Seem82, Seem90}.

In DNA tile self-assembly, the fundamental components are ``tiles'', which are comprised of interconnected DNA strands. Remarkably, these DNA tiles can be ``programmed'', via the careful configuration of their constituent DNA strands, to automatically coalesce into a desired target structure, the characteristics of which are completely determined by the ``programming'' of the DNA tiles. In order to fully realize the power of DNA tile self-assembly, we must study the algorithmic and mathematical underpinnings of tile self-assembly.

Perhaps the simplest mathematical model of algorithmic tile
self-assembly is Erik Winfree's abstract Tile Assembly Model (aTAM)
\cite{Winf98}. The aTAM is a deliberately over-simplified,
combinatorial model of nanoscale (DNA) tile self-assembly that
``effectivizes'' classical Wang tiling \cite{Wang61} in the sense that
the former augments the latter with a mechanism for sequential
``growth'' of a tile assembly. Very briefly, in the aTAM, the
fundamental components are un-rotatable, translatable square ``tile
types'' whose sides are labeled with (alpha-numeric) glue ``colors''
and (integer) ``strengths''. Two tiles that are placed next to each
other \emph{bind} if the glues on their abutting sides match in both
color and strength, and the common strength is at least a certain
(integer) ``temperature''.
Self-assembly starts from a ``seed'' tile
type, typically assumed to be placed at the origin of the coordinate
system, and proceeds nondeterministically and asynchronously as tiles
bind to the seed-containing assembly one at a time.

Despite its deliberate over-simplification, the aTAM is a
computationally expressive model. For example, Winfree \cite{Winf98}
proved that the model is Turing universal, which means that, in
principle, the process of self-assembly can be directed by any
algorithm. In this paper, we study the extent to which tile sets in
the aTAM can be algorithmically directed to ``strictly'' self-assemble
(i.e., place tiles at and only at locations that belong to) shapes
that are self-similar ``pier fractals''.

Intuitively, a ``pier fractal'' is a just-barely connected, self-similar fractal that contains the origin, as well as infinitely many, arbitrarily-large subsets of specially-positioned points that either lie at or ``on the far side'' from the origin of a pinch-point location (we make this notion precise in Section~\ref{sec:pier-fractals}). Note that ``tree'' and ``pinch-point'' fractals constitute notable, previously-studied sub-classes of pier fractals (e.g., see \cite{jSSADST, LutzShutters12, jSADSSF} for definitions).

There are examples of prior results related to the strict self-assembly of fractals in the aTAM. For example, Theorem 3.2 of \cite{jSSADST} bounds from below the size of the smallest tile set in which an arbitrary shape $X$ strictly self-assembles, by the depth of $X$'s largest finite sub-tree. Although not stated explicitly, an immediate corollary of this result is that no tree-fractal strictly self-assembles in the aTAM. In \cite{LutzShutters12}, Lutz and Shutters prove that a notable example of a tree fractal, the Sierpinski triangle, does not even ``approximately'' strictly self-assemble, in the sense that the discrete fractal dimension (see \cite{ZD}) of the symmetric difference of any set that strictly self-assembles and the Sierpinski triangle is at least that of the latter, which is approximately $\log_2 3$. Theorem 3.12 of \cite{jSADSSF}, the only prior result related to (the impossibility of) the strict self-assembly of pinch-point fractals, is essentially a qualitative generalization of Theorem 3.2 of \cite{jSSADST}.

While the strict self-assembly of certain classes of fractals in the
aTAM has been studied previously, \emph{nothing} is known about the
strict self-assembly in the aTAM of scaled-up versions of fractals,
where ``scaled-up'' means that each point in the original shape is
replaced by a $c \times c$ block of points, for some $c \in
\Z^+$. After all, certain classes of fractals defined by intricate
geometric properties, such as the existence of ``pinch-points'' or
``tree-ness'', are not closed under the scaling operation.
To see this, consider the full
connectivity graph of any shape in which each point in the shape is
represented by one vertex and edges exist between vertices that
represent adjacent points in the shape. If this graph is a tree and/or
contains one or more pinch-points, then the scaled-up version of the
original shape (with $c>1$) is not a tree and does not contain any
pinch points. This means that prior proof techniques that exploit
similar subtle geometric sub-structures of fractals (e.g.,
\cite{jSADSSF, jSSADST, LutzShutters12}) simply cannot be applied to
scaled-up versions of fractals. Thus, in this paper, we ask if it is
possible for a scaled-up version of a pier fractal to strictly
self-assemble in the aTAM.

The main contribution of this paper provides an answer to the previous
question, perhaps not too surprisingly to readers familiar with the
aTAM, in the negative: we prove that there is no pier fractal that
strictly self-assembles in the aTAM at any positive scale
factor. Furthermore, our definition of pier fractal includes, as a
strict subset, the set of all pinch-point fractals from
\cite{jSADSSF}. Our proof makes crucial use of a (modified version of
a) recent technical lemma developed by Meunier, Patitz, Summers,
Theyssier, Winslow and Woods \cite{WindowMovieLemma}, known as the
``Window Movie Lemma'' (WML). This (standard) WML is a kind of pumping
lemma for self-assembly since it gives a sufficient condition for
taking any pair of tile assemblies, at any temperature, and
``splicing'' them together to create a new valid tile assembly. Our
modified version of the WML, which we call the ``Closed Window Movie
Lemma'' (see Section~\ref{sec:window-movie-lemma} for a formal
statement and proof), allows one to replace a portion of a tile
assembly with another portion of the same assembly, assuming a certain
extra ``containment'' condition is met. Moreover, unlike in the
standard WML that lacks the extra containment assumptions, the
replacement of one part of the tile assembly with another in our
Closed WML only goes ``one way'', i.e., the part of the tile assembly
being used to replace another part cannot itself be replaced by the
part of the tile assembly it is replacing.


\section{Definitions}\label{sec-definitions}
In this section, we give a formal definition of Erik Winfree's
abstract Tile Assembly Model (aTAM), define ``pier fractals'' and
develop a ``Closed'' Window Movie Lemma.  

\subsection{Formal description of the abstract Tile Assembly Model}
\label{sec:tam-formal}
This section gives a formal definition of the abstract Tile Assembly Model (aTAM)~\cite{Winf98}. For readers unfamiliar with the aTAM,~\cite{Roth01} gives an excellent introduction to the model.

Fix an alphabet $\Sigma$.
$\Sigma^*$ is the set of finite strings over $\Sigma$. Let $\Z$, $\Z^+$, and $\N$ denote the set of integers, positive integers, and nonnegative integers, respectively. Given $V \subseteq \Z^2$, the \emph{full grid graph} of $V$ is the undirected graph $\fullgridgraph_V=(V,E)$,
such that, for all $\vec{x}, \vec{y}\in V$, $\left\{\vec{x},\vec{y}\right\} \in E \iff \| \vec{x} - \vec{y}\| = 1$, i.e., if and only if $\vec{x}$ and $\vec{y}$ are adjacent in the $2$-dimensional integer Cartesian space.

A \emph{tile type} is a tuple $t \in (\Sigma^* \times \N)^{4}$, e.g., a unit square, with four sides, listed in some standardized order, and each side having a \emph{glue} $g \in \Sigma^* \times \N$ consisting of a finite string \emph{label} and a nonnegative integer \emph{strength}.

We assume a finite set of tile types, but an infinite number of copies of each tile type, each copy referred to as a \emph{tile}. A tile set is a set of tile types and is usually denoted as $T$.

A {\em configuration} is a (possibly empty) arrangement of tiles on
the integer lattice $\Z^2$, i.e., a partial function $\alpha:\Z^2
\dashrightarrow T$.  Two adjacent tiles in a configuration
\emph{interact}, or are \emph{attached}, if the glues on their
abutting sides are equal (in both label and strength) and have
positive strength.  Each configuration $\alpha$ induces a
\emph{binding graph} $\bindinggraph_\alpha$, a grid graph whose
vertices are positions occupied by tiles, according to $\alpha$, with
an edge between two vertices if the tiles at those vertices
bind. An \emph{assembly} is a connected, non-empty configuration,
i.e., a partial function $\alpha:\Z^2 \dashrightarrow T$ such that
$\fullgridgraph_{\dom \alpha}$ is connected and $\dom \alpha \neq
\emptyset$.

Given $\tau\in\Z^+$, $\alpha$ is \emph{$\tau$-stable} if every cut-set
of~$\bindinggraph_\alpha$ has weight at least $\tau$, where the weight
of an edge is the strength of the glue it represents.\footnote{A
  \emph{cut-set} is a subset of edges in a graph which, when removed
  from the graph, produces two or more disconnected subgraphs. The
  \emph{weight} of a cut-set is the sum of the weights of all of the
  edges in the cut-set.} When $\tau$ is clear from context, we say
$\alpha$ is \emph{stable}.  Given two assemblies $\alpha,\beta$, we
say $\alpha$ is a \emph{subassembly} of $\beta$, and we write $\alpha
\sqsubseteq \beta$, if $\dom\alpha \subseteq \dom\beta$ and, for all
points $\vec{p} \in \dom\alpha$, $\alpha(\vec{p}) = \beta(\vec{p})$. For two
non-overlapping assemblies $\alpha$ and $\beta$, $\alpha \cup \beta$
is defined as the unique assembly $\gamma$ satisfying, for all
$\vec{x} \in \dom{\alpha}$, $\gamma(\vec{x}) = \alpha(\vec{x})$, for
all $\vec{x} \in \dom{\beta}$, $\gamma(\vec{x}) = \beta(\vec{x})$, and
$\gamma(\vec{x})$ is undefined at any point $\vec{x} \in \Z^2
\backslash \left( \dom{\alpha} \cup \dom{\beta} \right)$.

A \emph{tile assembly system} (TAS) is a triple $\mathcal{T} =
(T,\sigma,\tau)$, where $T$ is a tile set, $\sigma:\Z^2
\dashrightarrow T$ is the finite, $\tau$-stable, \emph{seed assembly},
and $\tau\in\Z^+$ is the \emph{temperature}.

Given two $\tau$-stable assemblies $\alpha,\beta$, we write $\alpha
\to_1^{\mathcal{T}} \beta$ if $\alpha \sqsubseteq \beta$ and
$|\dom\beta \setminus \dom\alpha| = 1$. In this case we say $\alpha$
\emph{$\mathcal{T}$-produces $\beta$ in one step}. If $\alpha
\to_1^{\mathcal{T}} \beta$, $ \dom\beta \setminus
\dom\alpha=\{\vec{p}\}$, and $t=\beta(\vec{p})$, we write $\beta =
\alpha + (\vec{p} \mapsto t)$.  The \emph{$\mathcal{T}$-frontier} of
$\alpha$ is the set $\partial^\mathcal{T} \alpha = \bigcup_{\alpha
  \to_1^\mathcal{T} \beta} (\dom\beta \setminus \dom\alpha$), i.e.,
the set of empty locations at which a tile could stably attach to
$\alpha$. The \emph{$t$-frontier} of $\alpha$, denoted
$\partial^\mathcal{T}_t \alpha$, is the subset of
$\partial^\mathcal{T} \alpha$ defined as
$\setr{\vec{p}\in\partial^\mathcal{T} \alpha}{\alpha \to_1^\mathcal{T}
  \beta \text{ and } \beta(\vec{p})=t}.$

Let $\mathcal{A}^T$ denote the set of all assemblies of tiles from $T$, and let $\mathcal{A}^T_{< \infty}$ denote the set of finite assemblies of tiles from $T$.
A sequence of $k\in\Z^+ \cup \{\infty\}$ assemblies $\alpha_0,\alpha_1,\ldots$ over $\mathcal{A}^T$ is a \emph{$\mathcal{T}$-assembly sequence} if, for all $1 \leq i < k$, $\alpha_{i-1} \to_1^\mathcal{T} \alpha_{i}$.
The {\em result} of an assembly sequence $\vec{\alpha}$, denoted as $\textmd{res}(\vec{\alpha})$, is the unique limiting assembly (for a finite sequence, this is the final assembly in the sequence).

We write $\alpha \to^\mathcal{T} \beta$, and we say $\alpha$ \emph{$\mathcal{T}$-produces} $\beta$ (in 0 or more steps), if there is a $\mathcal{T}$-assembly sequence $\alpha_0,\alpha_1,\ldots$ of length $k = |\dom\beta \setminus \dom\alpha| + 1$ such that
(1) $\alpha = \alpha_0$,
(2) $\dom\beta = \bigcup_{0 \leq i < k} \dom{\alpha_i}$, and
(3) for all $0 \leq i < k$, $\alpha_{i} \sqsubseteq \beta$.
If $k$ is finite then it is routine to verify that $\beta = \alpha_{k-1}$.

We say $\alpha$ is \emph{$\mathcal{T}$-producible} if $\sigma \to^\mathcal{T} \alpha$, and we write $\prodasm{\mathcal{T}}$ to denote the set of $\mathcal{T}$-producible assemblies. The relation $\to^\mathcal{T}$ is a partial order on $\prodasm{\mathcal{T}}$ \cite{Roth01,jSSADST}.

An assembly $\alpha$ is \emph{$\mathcal{T}$-terminal} if $\alpha$ is $\tau$-stable and $\partial^\mathcal{T} \alpha=\emptyset$.
We write $\termasm{\mathcal{T}} \subseteq \prodasm{\mathcal{T}}$ to denote the set of $\mathcal{T}$-producible, $\mathcal{T}$-terminal assemblies. If $|\termasm{\mathcal{T}}| = 1$ then  $\mathcal{T}$ is said to be {\em directed}.

We say that a TAS $\mathcal{T}$ \emph{strictly (or uniquely)
  self-assembles} $X \subseteq \Z^2$ if, for all $\alpha \in
\termasm{\mathcal{T}}$, $\dom\alpha = X$, i.e., if every terminal
assembly produced by $\mathcal{T}$ places a tile on every point in $X$
and does not place any tiles on points in $\Z^2 \backslash\, X$.

In this paper, we consider scaled-up versions of subsets of
$\Z^2$. Formally, if $X$ is a subset of $\Z^2$ and $c \in \Z^+$, then
a $c$-\emph{scaling} of $X$ is defined as the set $X^c = \left\{ (x,y)
\in \mathbb{Z}^2 \; \left| \; \left( \left\lfloor \frac{x}{c}
\right\rfloor, \left\lfloor \frac{y}{c} \right\rfloor \right) \in X
\right.\right\}$. Intuitively, $X^c$ is the subset of $\Z^2$ obtained
by replacing each point in $X$ with a $c \times c$ block of points. We
refer to the natural number $c$ as the \emph{scaling factor} or
\emph{resolution loss}.

\subsection{Pier fractals}\label{sec:pier-fractals}

In this section, we first introduce some terminology and then define a
class of fractals called ``pier fractals'' that is the focus of this
paper.

\begin{notation}
We use $\mathbb{N}_g$ to denote the subset $\{0, \ldots, g-1\}$ of
$\mathbb{N}$.
\end{notation}

\begin{notation}
If $A$ and $B$ are subsets of $\N^2$ and $k\in \N$, then $A+kB = \{\vec{m}+k\vec{n}~|~\vec{m}\in A$ and $\vec{n}\in B\}$.
\end{notation}

The following definition is a modification of Definition $2.11$ in \cite{jSADSSF}.

\begin{definition}\label{def:dssf}
Let $1<g \in \mathbb{N}$ and $\mathbf{X} \subset \mathbb{N}^2$. We say that
$\mathbf{X}$ is a \emph{$g$-discrete self-similar fractal} (or \emph{$g$-dssf}
for short), if there is a set $\{(0,0)\} \subset G \subset
\mathbb{N}_g^2$ with at least one point in every row and column, such
that $\displaystyle \mathbf{X} = \bigcup_{i=1}^{\infty}X_i$, where $X_i$, the
$i^{th}$ \emph{stage} of $\mathbf{X}$, is defined by $X_1= G$ and
$X_{i+1}= X_i\ +\ g^iG$. We say that $G$ is the
\emph{generator} of $\mathbf{X}$.
\end{definition}

Intuitively, a $g$-dssf is built as follows. Start by selecting points
in $\N_g^2$ satisfying the constraints listed in
Definition~\ref{def:dssf}. This first stage of the fractal is the
generator. Then, each subsequent stage of the fractal is obtained by
adding a full copy of the previous stage for every point in the
generator and translating these copies so that their relative
positions are identical to the relative positions of the individual
points in the gnerator.

\begin{definition}
Let $S$ be any finite subset of $\Z^2$. Let $l$, $r$, $b$, and
$t$ denote the following integers:

\centerline{$\displaystyle l_S = \min_{(x,y) \in S} x\qquad r_S =
  \max_{(x,y) \in S} x\qquad b_S = \min_{(x,y) \in S} y\qquad t_S =
  \max_{(x,y) \in S} y$}

An \emph{h-bridge} of $S$ is any subset of $S$ of the form
$hb_S(y) = \{(l_S,y),(r_S,y)\}$. Similarly, a \emph{v-bridge} of $S$ is any
subset of $S$ of the form $vb_S(x) = \{(x,b_S),(x,t_S)\}$. We say that a bridge is \emph{connected} if there is a simple path in $S$ connecting the two bridge points.
\end{definition}

\begin{notation}
Let $S$ be any finite subset of $\Z^2$. We will denote by
$nhb_S$ and $nvb_S$, respectively, the number of h-bridges and
the number of v-bridges of $S$.
\end{notation}

\ifabstract
\later{

\begin{definition}
If $G$ is the generator of any $g$-discrete self-similar
fractal, then the \emph{interior} of $G$ is $G \cap (\N_{g-1}\times \N_{g-1})$.
\end{definition}

\begin{lemma}
\label{lem:north-free-point}
Let $G$ be any finite subset of $\N^2$ that has at least one connected
h-bridge. If $G$ contains a connected component $C \subset G$ such
that $C \cap (\N \times \{t_G\}) \ne \emptyset$ and $C \cap (\{l_G\}
\times \N) = \emptyset$, then there exists a point $\vec{x}_N \in G
\backslash C$ such that $N\left(\vec{x}_N\right) \not \in G$ and
$\vec{x}_N \not \in \N \times \{t_G\}$.
\end{lemma}

\begin{proof}
Let $h$ be a connected h-bridge in $G$ and let $\pi$ be a connected
component in $G$ that contains a path connecting the two points in
$h$. Since $\pi$ connects the leftmost and rightmost columns of $G$
and $C$ does not contain any point in the leftmost column of $G$, $C
\cap \pi = \emptyset$. Since $C$ is a connected component that extends
vertically from row $t_C = t_G$ down to row $b_C$ and $C \cap \pi =
\emptyset$, $\pi$ must go around (and below) $C$. Furthermore, no
point in $C$ is adjacent to any point in $\pi$. Let $\vec{p}$ denote a
bottommost point $(x,b_C)$ in $C$, with $l_G < x \leq r_G$. Let
$\vec{q}$ denote the topmost point $(x,y)$ in $\pi \cap (\{x\}\times
\mathbb{N}_{b_C})$. Note that $\vec{p}$ and $\vec{q}$ are in the same
column and that $\vec{p}$ is above (but not adjacent to) $\vec{q}$, that
is, $y<b_C-1$.
Furthermore, $N(\vec{q}) \notin
G$. Since $\vec{q} \in \pi \subset G$ and $\vec{q} \not\in C$,
$\vec{q} \in G\backslash C$. Furthermore, since $\vec{q}=(x,y)$ and $y
< b_C -1 < b_C \leq t_C = t_G$, $\vec{q} \not \in \mathbb{N} \times
\{t_G\}$. In conclusion, $\vec{q}$ exists and is a candidate for the
role of $x_N$.
\end{proof}

\begin{lemma}
\label{lem:north-east-free-point}
Let $G$ be any finite subset of $\N^2$ that has at least one connected
v-bridge. If $G$ contains a connected component $C \subset G$ such
that $C \cap (\{r_G\} \times \N ) \ne \emptyset$, $C \cap (\N \times
\{t_G\}) \ne \emptyset$ and $C \cap (\N \times \{b_G\}) = \emptyset$,
then there exists a point $\vec{x}_{NE} \in G \backslash C$ such that
$E\left(\vec{x}_{NE}\right) \not \in G$, $\vec{x}_{NE} \in \N
\times \{t_G\}$ and $\vec{x}_{NE} \not \in \{r_G\} \times \N$.
\end{lemma}

\begin{proof}
Let $v$ be a connected v-bridge in $G$ and let $\pi$ be a connected
component in $G$ that contains a path connecting the two points in
$v$. Let $\pi_{t}$ denote the set $\pi \cap (\N \times
\{t_G\})$. Since this set cannot be empty, let us call its rightmost
point $\vec{p} = (x_{\pi},t_G)$. Similarly, let $C_{t}$ denote $C \cap (\N
\times \{t_G\})$. Since this set cannot be empty, let us call its leftmost point
$\vec{q} = (x_C,t_G)$.

Since $\pi$ connects the topmost and bottommost rows of $G$ and $C$
does not contain any point in the bottommost row of $G$, $C \cap \pi =
\emptyset$. This, together with the fact that $C$ contains a path from
the topmost row to the rightmost column of $G$ (that is, $C$ ``cuts
off'' the subset of $G$ that lies to the north-east of $C$ from the
rest of $G$), implies that each point in $\pi_{t}$ must appear
to the left of all the points in $C_{t}$, namely $x_{\pi} <
x_C$. In fact, since $\pi$ and $C$ cannot be connected,
$\vec{p}$ and $\vec{q}$ cannot be adjacent, i.e., $x_{\pi} <
x_C-1$. Therefore, $\vec{p}$ and $\vec{q}$ are both in the
topmost row of $G$ (thus $\vec{p} \in \N \times \{t_G\}$) and $\vec{p}$
is to the left of $\vec{q}$ (thus $\vec{p} \not \in \{r_G\} \times
\mathbb{N}$). Finally,  by construction, $\vec{p}\in G \backslash C$ and $E(\vec{p}) \not\in G$. In conclusion, $\vec{p}$ is a candidate for the role of $\vec{x}_{NE}$.
\end{proof}

\begin{lemma}
\label{lem:east-free-point}
Let $G$ be any finite subset of $\N^2$ that has at least one
connected v-bridge. If $G$ contains a connected component $C \subset
G$ such that $C \cap (\{r_G\} \times \N ) \ne \emptyset$ and $C \cap
(\N \times \{b_G\}) = \emptyset$, then there exists a point
$\vec{x}_{E} \in G \backslash C$ such that $E\left(\vec{x}_{E}\right)
\not \in G$ and $\vec{x}_{E} \not \in \{r_G\} \times \mathbb{N}$.
\end{lemma}

\begin{proof}
Let $v$ be a connected v-bridge in $G$ and let $\pi$ be a connected
component in $G$ that contains a path connecting the two points in
$v$. Since $\pi$ connects the topmost and bottommost rows of $G$ and
$C$ does not contain any point in the bottommost row of $G$, $C \cap
\pi = \emptyset$. Since $C$ is a connected component that extends
horizontally from column $l_C$ to column $r_C=r_G$ and $C \cap \pi =
\emptyset$, $\pi$ must go around (and to the left of)
$C$. Furthermore, no point in $C$ is adjacent to any point in $\pi$.
Let $\vec{p}$ denote a leftmost point $(l_C,y)$ in $C$, with $b_G < y
\leq t_G$. Let $\vec{q}$ denote the rightmost point $(x,y)$ in $\pi
\cap (\N_{l_C}\times \{y\})$. Note that $\vec{p}$ and $\vec{q}$ are in
the same row and that $\vec{q}$ is to the left of (but not adjacent to)
$\vec{p}$, that is, $x<l_C-1$. Furthermore, $E(\vec{q})
\notin G$. Since $\vec{q} \in \pi \subset G$ and $\vec{q} \not\in C$,
$\vec{q} \in G\backslash C$. Furthermore, since $\vec{q}=(x,y)$ and $x
< l_C -1 < l_C \leq r_C = r_G$, $\vec{q} \not \in \{r_G\}\times
\N$. In conclusion, $\vec{q}$ exists and is a candidate for the role
of $x_E$.
\end{proof}

\begin{lemma}
\label{lem:connected-implies-v-and-h-bridges}
Let $\mathbf{X} = \bigcup_{i=1}^{\infty}{X_i}$ be a $g$-discrete
self-similar fractal with generator $G$. If $\mathbf{X}$ is a tree,
then $G$ must have at least one connected h-bridge and at least one
connected v-bridge.
\end{lemma}

\begin{proof}
In this proof, we assume only that $G$ does not have
a connected h-bridge and reach a contradiction.  We omit the symmetric
reasoning that would allow us to prove that $G$ must contain at least
one connected v-bridge. Together, these two subproofs establish the
fact that $G$ must ontain at least one connected h-bridge and at least
one connected v-bridge.

Assume that $G$ does not have a connected h-bridge. We consider two
cases characterized  by the number of points in the leftmost column of $G$.

Case 1: $\left| G \cap \left(\{0\} \times \N\right)\right| = g$. Then,
the following three propositions hold:

(a) For every point $(1,y) \in G$, $N(1,y) \not \in G$. Indeed, if
$(1,y) \in G$ and $N(1,y) \in G$, then
$\{(0,y),N(0,y),N(1,y),(1,y)\} \subset \mathbf{X}$ would
constitute a cycle in $\mathbf{X}$, which contradicts the fact that
$\mathbf{X}$ is a tree.

(b) For every point $(1,y) \in G$, $S(1,y) \not \in G$. The justification is similar to the one for (a) above.

(c) $(1,g-1) \in G \Rightarrow (1,0) \not \in G$. Indeed, if $(1,0)
\in G$ and $(1,g-1) \in G$, then
$\{(0,g-1),N(0,g-1),N(1,g-1),(1,g-1)\} \subset \mathbf{X}$ would
constitute a cycle in $\mathbf{X}$, which contradicts the fact that
$\mathbf{X}$ is a tree.

We will now prove that there is no path in $\mathbf{X}$ from the
origin to any point $(x,y) \in \mathbf{X}$ with $x\geq 2g$. If there
were such a path $\pi$, it would include at least one pair of
consecutive points $(2g-1,y')$ and $(2g,y')$. Let us consider the
first such pair in $\pi$ and let $\lfloor \frac{y'}{g} \rfloor =
a$. Then $(2g-1,y') \in G+(g,ag)$. Since this copy of $G$ belongs to
the second column of copies of $G$ in $X_2$, we can use the
conjunction of propositions (a), (b) and (c) above to infer that
$\mathbf{X} \cap (G+(g,(a-1)g)=\emptyset$ and $\mathbf{X} \cap
(G+(g,(a+1)g)=\emptyset$. Therefore, $\pi$ must contain a sub-path
$\pi'$ from the leftmost column of $G+(g,ag)$ to $(2g-1,y')$, that is,
a path from $(g,y'')$ to $(2g-1,y')$, for $ag \leq y'' <(a+1)g$.  But
since the leftmost column of $G+(g,ag)$ contains $g$ points, there
must be a (vertical) path from $(g,y')$ to $(g,y'')$ fully contained
in the leftmost column of $G+(g,ag)$. Therefore, by concatenation of
this path to $\pi'$, $G+(g,ag)$ must contain a path from $(g,y')$ to
$(2g-1,y')$. But this path would be a connected h-bridge of $G
+(g,ag)$, which would imply that $G$ contains a connected h-bridge. So
we can conclude that there is no path in $\mathbf{X}$ from the origin
to any point east of the line $x=2g-1$. Since $\mathbf{X}$ contains an
infinite number of points in this region of $\N^2$, $\mathbf{X}$
cannot be connected, which is impossible since $\mathbf{X}$ is a
tree. This contradiction implies that $G$ must contain at least one
connected h-bridge.

Case 2: $\left| G \cap \left(\{0\} \times \N\right)\right| < g$. Since
$G$ does not have a connected h-bridge, one can show via a case
analysis that either $\mathbf{X}$ is disconnected or contains a
cycle. However, both of these scenarios are impossible since
$\mathbf{X}$ is a tree.

\end{proof}

\begin{notation}
Let $c, s\in \Z^+$ and $1<g \in \mathbb{N}$. Let $e, f \in
\N_g$. We use $S_s^c(e,f)$ to denote
$\{0,1,\ldots,cg^{s-1}-1\}^2+ cg^{s-1}(e,f)$.
\end{notation}

\begin{notation}
Let $1<g \in \N$. Let $\mathbf{X} =
\bigcup_{i=1}^{\infty}{X_i}$ be a $g$-discrete self-similar
fractal. If $s \in \Z^+$, we use $P_\mathbf{X}(s)$ to denote the
property: `` $X_s$ is a tree and $nhb_{X_s} = nvb_{X_s} = 1$''.

\end{notation}

\begin{lemma}\label{lem:inductive-step}
Let $1<g \in \mathbb{N}$. If $\mathbf{X}$ is a $g$-discrete self-similar fractal, then
$P_{\mathbf{X}}(i) \Rightarrow P_{\mathbf{X}}(i+1)$ for  $i \in
\mathbb{Z}^+$.
\end{lemma}

\begin{proof}
Let $\mathbf{X}$ be any $g$-discrete self-similar fractal. Let $i \in
\mathbb{Z}^+$. We will abbreviate $X_{i} \cap S^1_{i}(x,y)$
and $X_{i+1} \cap S^1_{i+1}(x,y)$ to $U(x,y)$ and $V(x,y)$,
respectively, where $x,y \in \mathbb{N}_g$. The definition of
$\mathbf{X}$ implies that the following proposition, which we refer to
as $(*)$, is true: ``Every non-empty $V$ subset of $X_{i+1}$
is a translated copy of $X_{i}$''.

Assume that $P_{\mathbf{X}}(i)$ holds.

First, we prove that $X_{i+1}$ is connected. Pick any two distinct
points $\vec{p}$ and $\vec{q}$ in $X_{i+1}$. If $\vec{p}$ and
$\vec{q}$ belong to the same $V$ subset of $X_{i+1}$, then there is a
simple path from $\vec{p}$ to $\vec{q}$ (because of $(*)$ and the fact
that $X_i$ is connected, by $P_{\mathbf{X}}(i)$). If $\vec{p}$ and $\vec{q}$
belong to two distinct $V$ subsets of $X_{i+1}$, say, $V(x_0,y_0)$ and
$V(x_k,y_k)$, then consider the corresponding two $U$ subsets
$U(x_0,y_0)$ and $U(x_k,y_k)$ of $X_i$, neither of which can be
empty. $P_{\mathbf{X}}(i)$ implies that there exists a simple path
from any point in $U(x_0,y_0)$ to any point in $U(x_k,y_k)$. Assume
that this path goes through the following sequence $P_i$ of $U$
subsets of $X_i$: $U(x_0,y_0), U(x_1,y_1), \ldots, U(x_{k-1},y_{k-1}),
U(x_k,y_k)$. $P_{\mathbf{X}}(i)$ and $(*)$
together imply that each one of the corresponding $V$ subsets of
$X_{i+1}$, i.e., $V(x_0,y_0)$, \ldots, $V(x_k,y_k)$, is connected and
contains a connected h-bridge and a connected v-bridge. Furthermore,
since any pair of consecutive $U$ subsets in $P_i$ are adjacent in
$X_i$, the same is true of the $V$ subsets of $X_{i+1}$ in the
sequence $P_{i+1}$: $V(x_0,y_0), V(x_1,y_1), \ldots,
V(x_{k-1},y_{k-1}), V(x_k,y_k)$. Since, for $i \in \mathbb{N}_{k}$,
$V(x_i,y_i)$ is adjacent to $V(x_{i+1},y_{i+1})$ and each one of these
subsets is connected and has at least one horizontal bridge and one
vertical bridge, there must be at least one simple path from any point
in $V(x_0,y_0)$ to any point in $V(x_k,y_k)$. Therefore, there exists
a simple path between $\vec{p} \in V(x_0,y_0)$ and $\vec{q} \in
V(x_k,y_k)$. Since this is true for any two distinct points $\vec{p}$ and
$\vec{q}$ in $X_{i+1}$, $X_{i+1}$ is connected.

Second, we prove that $nhb_{X_{i+1}} = nvb_{X_{i+1}}
= 1$.  Since the reasoning is similar for both horizontal and vertical
bridges, we only deal with $nhb_{X_{i+1}}$ here. By
$P_{\mathbf{X}}(i)$, $X_i$ contains exactly one horizontal
bridge. Therefore, there are exactly two subsets of $X_i$ of
the form $U(0,y)$ and $U(g-1,y)$, for some $y$ in $\mathbb{N}_g$, such
that there exist exactly two points $\vec{p}=(x_p,y_p)$ in $U(0,y)$ and
$\vec{q}=(x_q,y_q)$ in $U(g-1,y)$ with $y_p=y_q$. Now consider $V(0,y)$ and
$V(g-1,y)$. Since each one of these subsets of $X_{i+1}$ is a
translated copy of $X_i$, the westmost column of $V(0,y)$ is
identical to the westmost column of $X_i$ and the eastmost
column of $V(g-1,y)$ is identical to the eastmost column of
$X_i$. Therefore, the number of horizontal bridges in
$X_{i+1}$ that belong to $V(0,y) \cup V(g-1,y)$ is equal to
$nhb_{X_i}=1$. In other words,
$nhb_{X_{i+1}}\geq1$. Since both $X_{i}$ and
$X_{i+1}$ are built out of copies of their preceding stage
according to the same pattern (namely the generator of $\mathbf{X}$)
and we argued above that the only horizontal bridges in $X_i$
belong to $U(0,y) \cup U(g-1,y)$, the horizontal bridges
in $X_{i+1}$ can only belong to the subsets $V(0,y)$ and
$V(g-1,y)$. In other words, $nhb_{X_{i+1}}\leq 1$. Finally,
$nhb_{X_{i+1}}=1$.

Third, we prove that $X_{i+1}$ is acyclic. For the sake of obtaining a
contradiction, assume that there exists a simple cycle $C$ in
$X_{i+1}$. Let the sequence $P_{i+1}$ of adjacent $V$ subsets that $C$
traverses be $V(x_0,y_0)$, \ldots, $V(x_k,y_k)$. If $P_{i+1}$ has
length one, then $C$ is contained in a single (translated) copy of
$X_{i}$ (by $(*)$), which contradicts the fact that $X_{i}$ is acyclic
(by $P_{\mathbf{X}}(i)$). Otherwise, $C$ traverses all of the $V$
subsets in $P_{i+1}$, whose length is at least two.  Following the
same reasoning as above, there must exist a corresponding sequence
$P_i$, namely $U(x_0,y_0)$, \ldots, $U(x_k,y_k)$, of $U$ subsets in
$X_i$.  Since each subset in this sequence is connected, contains one
horizontal bridge and one vertical bridge (by $P_{\mathbf{X}}(i)$),
and is adjacent to its neighbors in the sequence, the union of these
subsets forms a connected component that must contain at least one
simple cycle, which contradicts the fact that $X_i$ is a tree (by
$P_{\mathbf{X}}(i)$). In all cases, we reached a
contradiction. Therefore, $X_{i+1}$ is acyclic.

Finally, since $X_{i+1}$ is a tree and
$nhb_{X_{i+1}}=nvb_{X_{i+1}}=1$, $P_{\mathbf{X}}(i+1)$ holds.

\end{proof}

\setcounter{theorem}{1}
\begin{theorem}
$\displaystyle\mathbf{T} = \bigcup_{i=1}^{\infty}{T_i}$ is a $g$-discrete self-similar tree fractal, for some $g > 1$, with generator $G$ if and only if \\
\hspace*{1cm}a. $G$ is a tree, and\\
\hspace*{1cm}b. $nhb_G = nvb_G =1$
\end{theorem}

\begin{proof}
Assume that $\mathbf{T}$ is a $g$-discrete self-similar tree fractal
with generator $G$. Thus, $\mathbf{T}$ is acyclic and connected. If
$nhb_G < 1$ or $nvb_G < 1$, then $\mathbf{T}$ is trivially
disconnected. Thus, $nhb_G \geq 1$, $nvb_G \geq 1$.

Since $\mathbf{T}$ is acyclic, $G$ must be acyclic as well, for if $G$ were not acyclic, then $\mathbf{T}$ would not be, as $G \subset \mathbf{T}$.

We will now show that $G$ is connected. To see this, assume that $G$
is disconnected. First, note that, if $G$ has a connected component
contained strictly within its interior, then $\mathbf{T}$ is trivially
disconnected.

Second, if $G$ is disconnected, then $G$ contains a connected
component that touches at most two sides of $G$. To see this, note
that Lemma~\ref{lem:connected-implies-v-and-h-bridges} says that $G$
has at least one connected h-bridge and at least one connected
v-bridge. If $G$ had a connected component, say $C$, that touched
three or more sides of $G$, then due to the existence of at least one
connected h-bridge and at least one connected v-bridge, $G$ would
necessarily have another connected component, say $C'$, that could
only touch at most two sides of $G$.

We now proceed with a case analysis based on the number of sides of
$G$ that the connected component touches (one or two sides) and
the relative positions of these sides (adjacent or opposite sides).

Case 1: Assume that $G$ has a connected component, say $C$, that does
not contain the origin but does contain points in the northmost row
and eastmost column of $G$ (and there is no path in $G$ from the
origin to any point in $C$). We will call this case
``NE''. Lemma~\ref{lem:connected-implies-v-and-h-bridges} says that
$G$ has at least one connected h-bridge and at least one connected
v-bridge. Therefore, Lemma~\ref{lem:north-free-point} says that $G$
has a north-free point not in the northmost row, say $\vec{x}_{N}$,
and Lemma~\ref{lem:north-east-free-point} says that $G$ has an
east-free point in the northmost row but not in the eastmost column,
say $\vec{x}_{NE}$. Let $C' = C + g^2 \vec{x}_N + g
\vec{x}_{NE}$. Since $\vec{x}_N$ is north-free and not in the
northmost row of $G$, $N\left(\vec{x}_N\right) \not \in G$, whence
$\mathbf{T} \cap \left(\left\{0,\ldots, g^2-1\right\}^2 + g^2
N\left(\vec{x}_N\right)\right) = \emptyset$.  Since $\vec{x}_{NE}$ is
in the northmost row, this means the northmost point in every column
of $C'$ is north-free in $\mathbf{T}$. Since $\vec{x}_{NE}$ is not in
the eastmost column of $G$, $E\left(\vec{x}_{NE}\right) \not \in G$,
whence $\mathbf{T} \cap \left( g^2 \vec{x}_N + \left(
\left\{0,\ldots,g-1\right\}^2 + g E\left(\vec{x}_{NE}\right) \right)
\right) = \emptyset$. This means that the eastmost point in every row
of $C'$ is east-free in $\mathbf{T}$. We also know that the westmost
point in every row of $C$ is west-free in $G$ and the southmost point
in every column of $C$ is south-free in $G$, therefore the westmost
point in every row of $C'$ is west-free in $\mathbf{T}$ and the
southmost point in every column of $C'$ is south-free in
$\mathbf{T}$. Thus, there is no path in $\mathbf{T}$ from any point in
$C'$ to the origin, which contradicts the assumption that $\mathbf{T}$
is connected. The ``NW'' and ``SE'' cases can be handled with a
similar argument. Note that, in the ``SW'' case, the connected
component is contained strictly within the interior of the
generator. Such situations were handled above.

Case 2: Assume that $G$ has a connected component, say $C$, that
contains points in the eastmost column but does not contain the origin
nor points in the westmost column of $G$ nor the northmost or
southmost rows of $G$. This is the ``E'' case. In this case,
Lemma~\ref{lem:east-free-point} says that there is an east-free point
in $G$ that is not in the eastmost column of $G$. Call this point
$\vec{x}_{E}$ and define $C' = C+ g\vec{x}_E$. Following directly from
the definition of the ``E'' case, the northmost point in every column
of $C$ is north-free in $G$, the southmost point in every column of
$C$ is south-free in $G$ and the westmost point in every row of $C$ is
west-free in $G$. From the definition of $C'$ and the fact that
$\vec{x_E}$ is east-free, it follows that the eastmost (respectively,
westmost) point in every row of $C'$ is east-free (respectively,
west-free) in $\mathbf{T}$. Similarly, the northmost (respectively,
southmost) point in every column of $C'$ is north-free (respectively,
south-free) in $\mathbf{T}$. Therefore, there is no path in
$\mathbf{T}$ from any point in $C'$ to the origin, which contradicts
the assumption that $\mathbf{T}$ is connected.  The ``N'' case can be
handled with a similar argument. Note that, in the ``W'' and ``S''
cases, the connected component is contained strictly within the
interior of the generator. Such situations were handled above.

Case 3: Assume that $G$ has a connected component, say $C$, that
contains points in both the eastmost and westmost columns of
$G$. This is the ``EW'' case. In this case, since $G$ contains at
least one connected v-bridge, $C$ must contain all connected v-bridges
of $G$ (since $C$ must have a non-empty intersection with each
connected v-bridge in $G$). Therefore, $C$ touches all four sides of
$G$. If $C$ contains the origin, then there must exist another
disconnected component, say $C'$, that does not contain the origin and
$C'$ must belong to one of the previous cases. If $C$ does not contain
the origin, then the origin itself must be part of a connected
component that is not connected to $C$ nor to any other point in
$\mathbf{T}$, which contradicts the assumption that $\mathbf{T}$ is
connected.  The ``NS'' case can be handled with a similar argument.

Therefore, in all cases, $G$ is connected. Since we argued above that
$G$ is acyclic, we may conclude that $G$ is a tree.

Finally, since $G$ is a tree, it must be the case that $nvb_G \leq 1$ and $nhb_G \leq 1$, otherwise $\mathbf{T}$ would contain a cycle, whence $nvb_G = nhb_G = 1$.

Now we prove that if $G$ is a tree and $nhb_G = nvb_G = 1$, then
$\mathbf{T}$ is a tree.

Assume that $G$ is a tree and $nhb_G = nvb_G = 1$. Then
$P_{\mathbf{T}}(1)$ holds (since $G=T_1)$. Furthermore, by
Lemma~\ref{lem:inductive-step}, $P_{\mathbf{T}}(i) \Rightarrow
P_{\mathbf{T}}(i+1)$ for $i \in \Z^+$. Thus, by induction,
$P_{\mathbf{T}}(i)$ holds for $i \in \Z^+$, which implies that
each stage in $\mathbf{T}$ is a tree. We now prove that $\mathbf{T}$ is a tree.

First, $\mathbf{T}$ is connected, since each stage of $\mathbf{T}$ is connected.

Second, we prove that $\mathbf{T}$ cannot contain a cycle. Assume, for
the sake of obtaining a contradiction, that there exist two distinct
points $\vec{p}$ and $\vec{q}$ in $\mathbf{T}$ such that there exist
two distinct simple paths from $\vec{p}$ to $\vec{q}$. Since both of
these paths must be finite, the cycle that they form must also be
finite. Therefore, this cycle must be fully contained in some stage of
$\mathbf{T}$, which contradicts the fact that all stages of
$\mathbf{T}$ are trees.

In conclusion, $\mathbf{T}$ is connected and acyclic, and is thus a tree.
\end{proof}\

}
\fi

\begin{notation}
The directions $\mathcal{D} = \{N,E,S,W\}$ will be used as functions
from $\mathbb{Z}^2$ to $\mathbb{Z}^2$ such that $N(x,y) = (x,y+1)$,
$E(x,y) = (x+1,y)$, $S(x,y) = (x,y-1)$ and $W(x,y) = (x-1,y)$. Note
that $N^{-1} = S$ and $W^{-1}=E$.
\end{notation}

\begin{notation}
Let $X \subseteq \mathbb{Z}^2$. We say that a point $(x,y) \in X$ is
$D$-\emph{free} in $X$, for some direction $D \in \mathcal{D}$, if
$D(x,y) \not \in X$.
\end{notation}

\begin{definition}
Let $G$ be the generator of any $g$-discrete self-similar fractal.  A
\emph{pier} is a point in $G$ that is $D$-free in $G$ for exactly
three of the four directions in $\mathcal{D}$. We say that a pier
$(x,y)$ is \emph{$D$-pointing} (or \emph{points $D$}) if $D^{-1}(x,y)
\in G$. Note that a pier always points in exactly one direction.
\end{definition}

\begin{definition}
Let $G$ be the generator of any $g$-discrete self-similar fractal with
exactly one h-bridge and one v-bridge.  $G$ may contain up to four
distinct types of piers characterized by the number of bridges they belong
to. Each pier may belong to no more than two bridges.  A \emph{real
  pier} is a pier that does not belong to any bridge in $G$. A
\emph{single-bridge pier} belongs to exactly one bridge. A
\emph{double-bridge pier} belongs to exactly two bridges. Finally, we
will distinguish between two sub-types of single-bridge piers. If the
pier is pointing in a direction that is parallel to the direction of
the bridge (i.e., if the pier points north or south and belongs to a
v-bridge, or the pier points east or west and belongs to an h-bridge),
the pier is a \emph{parallel single-bridge pier}. If the pier is
pointing in a direction that is orthogonal to the direction of the
bridge (i.e., if the pier points north or south and belongs to an
h-bridge, or the pier points east or west and belongs to a v-bridge),
the pier is an \emph{orthogonal single-bridge pier}.

\end{definition}

For example, the generator in Figure~\ref{fig:piers} below contains
the h-bridge $\{(0,0),(4,0)\}$ and the v-bridge $\{(4,0),(4,4)\}$. The
point $(1,4)$ is a real pier. The point $(0,0)$
is an orthogonal single-bridge pier. The point $(4,4)$
is a parallel single-bridge pier. The point $(4,0)$ is
a double-bridge pier.

\begin{figure}[h]
\centering
 \includegraphics[width=0.4\textwidth]{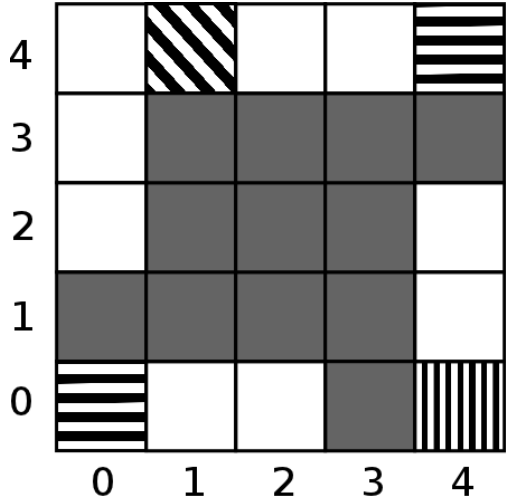}
\caption{A $5\times 5$ generator containing one h-bridge, one
  v-bridge and four piers.}
\label{fig:piers}
\end{figure}

We are now ready to define the class of fractals that is the main
focus of this paper.

\begin{definition}\label{def:pier-fractal}
$\mathbf{P}$ is a \emph{pier fractal} if and only if $\mathbf{P}$ is a discrete self-similar fractal with generator $G$ such that:\\
\hspace*{15pt}a. The full grid graph of $G$ is connected, and\\
\hspace*{15pt}b. $nhb_G = nvb_G =1$, and\\
\hspace*{15pt}c. $G$ contains at least one pier.
\end{definition}

\subsection{The Closed Window Movie Lemma}\label{sec:window-movie-lemma}
In this subsection, we develop a more accommodating (modified) version
of the standard Window Movie Lemma (WML) \cite{WindowMovieLemma}. Our
version of the WML, which we call the ``Closed Window Movie Lemma'',
allows us to replace one portion of a tile assembly with another,
assuming certain extra ``containment'' conditions are met. Moreover,
unlike in the standard WML that lacks the extra containment
assumptions, the replacement of a portion of one tile assembly
with another portion of the same assembly in our
Closed WML only goes ``one way'', i.e., the part of the tile assembly
being used to replace another part cannot itself be replaced by the
part of the tile assembly it is replacing. We must first define some
notation that we will use in our Closed Window Movie Lemma.

A \emph{window} $w$ is a set of edges forming a cut-set of the full grid
graph of $\mathbb{Z}^2$. For the purposes of this paper, we say that a
\emph{closed window} $w$ induces a cut\footnote{A \emph{cut} is a
  partition of the vertices of a graph into two disjoint subsets that
  are joined by at least one edge.} of the full grid graph of
$\mathbb{Z}^2$, written as $C_w = (C_{<\infty},C_\infty)$, where
$C_{\infty}$ is infinite, $C_{<\infty}$ is finite and for all pairs of
points $\vec{x},\vec{y} \in C_{<\infty}$, no simple path connecting
$\vec{x}$ and $\vec{y}$ in the full grid graph of $C_{<\infty}$
crosses the cut $C_w$. We call the set of vertices that make up
$C_{<\infty}$ the \emph{inside} of the window $w$, and write
$inside(w) = C_{<\infty}$ and $outside(w) = \mathbb{Z}^2 \backslash\,inside(w) =
C_\infty$. We say that a window $w$ is \emph{enclosed} in another
window $w'$ if $inside(w) \subseteq inside(w')$.

Given a window $w$ and an assembly $\alpha$, a window that {\em
  intersects} $\alpha$ is a partitioning of $\alpha$ into two
configurations (i.e., after being split into two parts, each part may
or may not be disconnected). In this case we say that the window $w$
cuts the assembly $\alpha$ into two configurations $\alpha_L$
and~$\alpha_R$, where $\alpha = \alpha_L \cup \alpha_R$. If $w$ is a
closed window, for notational convenience, we write $\alpha_I$ for the
configuration inside $w$ and $\alpha_O$ for the configuration outside $w$.
Given a window $w$, its translation by a vector $\vec{c}$, written $ w
+ \vec{c}$ is simply the translation of each of $w$'s elements (edges)
by~$\vec{c}$.


For a window $w$ and an assembly sequence $\vec{\alpha}$, we define a window movie~$M$ to be the order of placement, position and glue type for each glue that appears along the window $w$ in $\vec{\alpha}$. Given an assembly sequence $\vec{\alpha}$ and a window $w$, the associated {\em window movie} is the maximal sequence $M_{\vec{\alpha},w} = (v_{0}, g_{0}) , (v_{1}, g_{1}), (v_{2}, g_{2}), \ldots$ of pairs of grid graph vertices $v_i$ and glues $g_i$, given by the order of the appearance of the glues along window $w$ in the assembly sequence $\vec{\alpha}$.
Furthermore, if $k$ glues appear along $w$ at the same instant (this happens upon placement of a tile which has multiple  sides  touching $w$) then these $k$ glues appear contiguously and are listed in lexicographical order of the unit vectors describing their orientation in $M_{\vec{\alpha},w}$.

Let $w$ be a window and $\vec{\alpha}$ be an assembly sequence and $M = M_{\vec{\alpha},w}$. We use the notation $\mathcal{B}\left(M\right)$ to denote the \emph{bond-forming submovie} of $M$, i.e., a restricted form of $M$, which consists of only those steps of $M$ that place glues that eventually form positive-strength bonds in the assembly $\alpha = \res{\vec{\alpha}}$. Note that every window movie has a unique bond-forming submovie.

\begin{lemma}[Closed Window Movie Lemma]\label{lem:wml}
Let $\vec{\alpha}~=~(\alpha_i~|~0\leq~i<~l)$, with $l \in \mathbb{Z}^+
\cup \{\infty\}$, be an assembly sequence in some TAS $\mathcal{T}$
with result $\alpha$. Let $w$ be a closed window that partitions
$\alpha$ into $\alpha_I$ and $\alpha_O$, and $w'$ be a closed window
that partitions $\alpha$ into $\alpha_I'$ and $\alpha_O'$.  If
$\mathcal{B}(M_{\vec{\alpha},w}) + \vec{c} =
\mathcal{B}(M_{\vec{\alpha},w'})$ for some $\vec{c}\neq(0,0)$ and
the window $w+\vec{c}$ is enclosed in $w'$, then the assembly $\alpha'_O \cup (\alpha_I + \vec{c})$ is
in $\mathcal{A[\mathcal{T}]}$.
\end{lemma}

\begin{proof}
Before we proceed with the proof, the next paragraph introduces some
notation taken directly from \cite{WindowMovieLemma}.

For an assembly sequence $\vec{\alpha} = (\alpha_i \mid 0 \leq i <
l)$, we write $\left| \vec{\alpha} \right| = l$ (note that if
$\vec{\alpha}$ is infinite, then $l = \infty$). We write
$\vec{\alpha}[i]$ to denote $\vec{x} \mapsto t$, where $\vec{x}$
and~$t$ are such that $\alpha_{i+1} = \alpha_i + \left(\vec{x} \mapsto
t\right)$, i.e., $\vec{\alpha}[i]$ is the placement of tile type $t$
at position~$\vec{x}$, assuming that $\vec{x} \in
\partial_{t}\alpha_i$. We write $\vec{\alpha}[i] + \vec{c}$, for some
vector $\vec{c}$, to denote $\left(\vec{x}+\vec{c}\right) \mapsto
t$. We define $\vec{\alpha} = \vec{\alpha} + \left(\vec{x} \mapsto
t\right) = (\alpha_i \mid 0 \leq i < k + 1)$, where $\alpha_{k} =
\alpha_{k-1} + \left(\vec{x} \mapsto t\right)$ if $\vec{x} \in
\partial_{t}\alpha_{k-1}$ and undefined otherwise, assuming $\left|
\vec{\alpha} \right| > 0$. Otherwise, if $\left| \vec{\alpha} \right|
= 0$, then $\vec{\alpha} = \vec{\alpha} + \left(\vec{x} \mapsto t
\right) = (\alpha_0)$, where $\alpha_0$ is the assembly such that
$\alpha_0\left(\vec{x}\right) = t$ and is undefined at all other
positions. This is our notation for appending steps to the assembly
sequence $\vec{\alpha}$: to do so, we must specify a tile type $t$ to
be placed at a given location $\vec{x} \in \partial_t\alpha_{i}$. If
$\alpha_{i+1} = \alpha_i + \left(\vec{x} \mapsto t\right)$, then we
write $Pos\left(\vec{\alpha}[i]\right) = \vec{x}$ and
$Tile\left(\vec{\alpha}[i]\right) = t$. For a window movie $M =
(v_0,g_0), (v_1,g_1), \ldots$, we write $M[k]$ to be the pair
$\left(v_{k},g_{k}\right)$ in the enumeration of $M$ and
$Pos\left(M[k]\right) = v_{k}$, where $v_{k}$ is a vertex of a grid
graph.

We now proceed with the proof, throughout which we will assume that $M = \mathcal{B}\left(M_{\vec{\alpha},w}\right)$ and $M' = \mathcal{B}\left(M_{\vec{\alpha},w'}\right)$. Since $M + \vec{c} =
M'$ for some $\vec{c}\neq(0,0)$ and $w$ and $w'$ are both closed windows, it must be the case that the seed tile of $\alpha$ is in $\dom{\alpha_O} \cap \dom{\alpha'_O}$ or in $\dom{\alpha_I} \cap \dom{\alpha'_I}$. In other words, the seed tile cannot be in $\dom{\alpha_I} \backslash\, \dom{\alpha'_I}$ nor in $\dom{\alpha'_I} \backslash\, \dom{\alpha_I}$. Therefore, assume without loss of generality that the seed tile is in $\dom{\alpha_O} \cap \dom{\alpha'_O}$.

The algorithm in Figure \ref{fig:algo-seq} describes how to produce a new valid assembly sequence $\vec{\gamma}$.

%
%
%
%
%
%
%
%
%
%

\begin{figure}[htp]
\centering
\begin{minipage}{0.5\linewidth}
\begin{algorithm}[H]
\SetAlgoLined
Initialize $i$, $j= 0$ and $\vec{\gamma}$ to be empty

\For{$k = 0$ \KwTo $|M| - 1$}{
  \If{$Pos(M'[k]) \in \dom{\alpha'_O}$}{
    \While{$Pos(\vec{\alpha}[i])\neq Pos(M'[k])$}{
      \If{$Pos(\vec{\alpha}[i]) \in \dom{\alpha'_O}$}{$\vec{\gamma} = \vec{\gamma} + \vec{\alpha}[i]$}
      $i = i + 1$
    }
    $\vec{\gamma} = \vec{\gamma} + \vec{\alpha}[i]$

    $i = i + 1$
  }
  \Else {
    \While{$Pos(\vec{\alpha}[j])\neq Pos(M[k])$}{
      \If{$Pos(\vec{\alpha}[j]) \in \dom{\alpha_I}$}{
        $\vec{\gamma} = \vec{\gamma} + \left(\vec{\alpha}[j] + \vec{c}\right)$}

        $j = j + 1$
    }
    $\vec{\gamma} = \vec{\gamma} + \left(\vec{\alpha}[j] + \vec{c}\right)$

    $j = j + 1$

  }

}

\While{$inside(w) \cap \partial \res{\vec{\gamma}} \ne \emptyset$}{
    \If{$Pos(\vec{\alpha}[j]) \in \dom{\alpha_I}$}{$\vec{\gamma} = \vec{\gamma} + (\vec{\alpha}[j] + \vec{c})$}

    $j = j + 1$
    }

\While {$i < |\vec{\alpha}|$}{
  \If{$Pos(\vec{\alpha}[i]) \in \dom{\alpha'_O}$}{$\vec{\gamma} = \vec{\gamma} + \vec{\alpha}[i]$}

    $i = i + 1$
}

\Return $\vec{\gamma}$

\end{algorithm}
\end{minipage}
\caption{The algorithm to produce a valid assembly sequence $\vec{\gamma}$.}
\label{fig:algo-seq}

\end{figure}

If we assume that  the assembly sequence $\vec{\gamma}$ ultimately produced by the algorithm is valid, then the result of $\vec{\gamma}$ is indeed $\alpha'_O \cup \left( \alpha_I + \vec{c} \right)$. Observe that $\alpha_I$ must be finite, which implies that $M$ is finite. If $|\vec{\alpha}| < \infty$, then all loops will terminate. If $|\vec{\alpha}| = \infty$, then $|\alpha'_O| = \infty$ and the first two loops will terminate and the last loop will run forever. In either case, for every tile in~$\alpha'_O$ and $\alpha_I + \vec{c}$, the algorithm adds a step to the sequence $\vec{\gamma}$ involving the addition of this tile to the assembly. However, we need to prove that the assembly sequence $\vec{\gamma}$ is valid. It may be the case that either: 1. there is insufficient bond strength between the tile to be placed and the existing neighboring tiles, or 2. a tile is already present at this location.

\textbf{Case 1:}
In this case, we claim the following: at each step of the algorithm, the current version of $\vec{\gamma}$ is a valid assembly sequence whose result is a producible subassembly of $\alpha'_O \cup \left(\alpha_I + \vec{c}\right)$. Note that the three loops in the algorithm iterate through all steps of $\vec{\alpha}$, such that, when adding $\vec{\alpha}[i]$ (or $\vec{\alpha}[j] + \vec{c}$) to $\vec{\gamma}$, all steps of the window movie corresponding to the positions/glues of tiles to which $\vec{\alpha}[i]$ (or $\vec{\alpha}[j] + \vec{c}$) initially bind in $\vec{\alpha}$ have occurred. In other words, when adding $\vec{\alpha}[i]$ (or $\vec{\alpha}[j] + \vec{c}$) to $\vec{\gamma}$, the tiles to which $\vec{\alpha}[i]$ (or $\vec{\alpha}[j] + \vec{c}$) initially bind have already been added to $\vec{\gamma}$ by the algorithm. Similarly, all tiles in $\alpha'_O$ (or $\alpha_I + \vec{c}$) added to $\alpha$ before step $i$ (or $j$) in the assembly sequence $\vec{\alpha}$ have already been added to $\vec{\gamma}$.

So, if the tile $Tile\left(\vec{\alpha}[i]\right)$ that is added to the subassembly of $\alpha$ produced after $i-1$ steps can bond at a location in $\alpha'_O$ to form a $\tau$-stable assembly, then the same tile added to the result of $\vec{\gamma}$, which is producible, must also bond to the same location in the result of $\vec{\gamma}$, as the neighboring glues consist of (i) an identical set of glues from tiles in the subassembly of $\alpha'_O$ and (ii) glues on the side of the window movie containing~$\alpha_I + \vec{c}$.  Similarly, the tiles of $\alpha_I + \vec{c}$ must also be able to bind.

\textbf{Case 2:} Since we only assume that $\mathcal{B}\left(M_{\vec{\alpha},w}\right) + \vec{c} = \mathcal{B}\left(M_{\vec{\alpha},w'}\right)$, as opposed to the stronger condition $\mathcal{B}\left(M_{\vec{\alpha},w+\vec{c}}\right) = \mathcal{B}\left(M_{\vec{\alpha},w'}\right)$, which is assumed in the standard WML, we must show that $\dom{\left(\alpha_I + \vec{c}\right)} \cap \dom{\alpha'_O} = \emptyset$. To see this, observe that, by assumption, $w + \vec{c}$ is enclosed in $w'$, which, by definition, means that $inside\left(w+\vec{c}\right) \subseteq inside(w')$. Then we have $\vec{x} \in \dom{\alpha'_O} \Rightarrow \vec{x} \in outside(w') \Rightarrow \vec{x} \not \in inside\left(w'\right) \Rightarrow \vec{x} \not \in inside\left(w+\vec{c}\right) \Rightarrow \vec{x} \not \in \dom{\left(\alpha_I + \vec{c}\right)}$. Thus, locations in $\alpha_I + \vec{c}$ only have tiles from $\alpha_I$ placed in them, and locations in $\alpha'_O$ only have tiles from $\alpha'_O$ placed in them.

So the assembly sequence of $\vec{\gamma}$ is valid, i.e., every single-tile addition in $\vec{\gamma}$ adds a tile to the assembly to form a new producible assembly. Since we have a valid assembly sequence, as argued above, the resulting producible assembly is~$\alpha'_O \cup \left(\alpha_I + \vec{c}\right)$.

\end{proof}

\section{Scaled pier fractals do not strictly self-assemble in the aTAM}

In this section, we first define some notation and establish
preliminary results. Then we prove our main result, namely
that no scaled pier fractal self-assembles in the
aTAM. Finally, we prove corollaries of our main result, including the
fact that no scaled tree fractal self-assembles in the aTAM.

\subsection{Preliminaries}

Recall that each stage $X_s$ ($s>1$) of a $g$-dssf (scaled by a factor
$c$) is made up of copies of the previous stage $X_{s-1}$, each of
which is a square of size $cg^{s-1}$.

In the proof of our main result, we will need to refer to one of
the squares of size $cg^{s-2}$ inside the copies of stage $X_{s-1}$, leading to
the following notation.

\begin{notation}\label{smallwindows}
Let $c \in \Z^+$, $1<s \in \N$ and $1<g \in
\N$. Let $e, f, p, q \in \N_g$. We use
$S_s^c(e,f,p,q)$ to denote $\{0,1,\ldots,cg^{s-2}-1\}^2+ cg^{s-1}(e,f)
+ cg^{s-2}(p,q)$ and
$W_s^c(e,f,p,q)$ to denote the square-shaped, closed window whose inside
is $S_s^c(e,f,p,q)$.
\end{notation}

In Figure~\ref{fig:case2example} below, the bottom and top (circular) magnifications show
the windows $W_2^1(0,1,3,2)$ and $W_3^1(0,1,3,2)$, respectively.

Next, we will need to translate a small window to a position inside a
larger window. These two windows will correspond to squares at the
same relative position in different stages $i$ and $j$ of a $g$-dssf.

\begin{notation}\label{translations}
Let $c \in \Z^+$,  $1<i \in \N$, $1<j \in \N$, with $i<j$, and $e, f, xsp, q\in\N_g$.
We use $\vec{t}^c_{i\rightarrow j}(e,f,p,q)$ to denote
the vector joining the southwest corner of $W_i^c(e,f,p,q)$ to the
southwest corner of $W_j^c(e,f,p,q)$. In other words, $\vec{t}^c_{i\rightarrow
  j}(e,f,p,q) = \left(c\left(g^{j-1}-g^{i-1}\right)e+c\left(g^{j-2}-g^{i-2}\right)p,
c\left(g^{j-1}-g^{i-1}\right)f+c\left(g^{j-2}-g^{i-2}\right)q\right)$.
\end{notation}

For example,
in Figure~\ref{fig:case2example} below,
$\vec{t}_{2\rightarrow 3}^1(0,1,3,2) = (9,18)$.

To apply Lemma~\ref{lem:wml}, we will need the bond-forming submovies
to line up. Therefore, once the two square windows share their
southwest corner after using the translation defined above, we will
need to further translate the smallest one either up or to the right,
or both, depending on which side of the windows contains the
bond-forming glues, which, in the case of scaled pier fractals, always
form a straight (vertical or horizontal) line of length $c$. We will
compute the coordinates of this second translation in our main
proof. For now, we establish an upper bound on these coordinates that
will ensure that the translated window will remain enclosed in the
larger window.

\begin{lemma}\label{lem:enclosure2}
Let $c \in \mathbb{Z}^+$, $1<i \in \N$, $1<j \in \N$, with $i<j$,
$e, f,p,q \in \mathbb{N}_g$, and $x,y \in \mathbb{N}$. Let $m =
c(g^{j-2}-g^{i-2})$. If $x\leq m$ and $y\leq m$, then the window
$W^c_i(e,f,p,q) +\vec{t}^c_{i\rightarrow j}(e,f,p,q)+(x,y)$ is
enclosed in the window $W^c_j(e,f,p,q)$.
\end{lemma}

\begin{proof}
Let $W$ and $w$ denote $W^c_j(e,f,p,q)$ and
$W^c_i(e,f,p,q)+\vec{t}^c_{i\rightarrow j}(e,f,p,q)$,
respectively. Since $W$ and $w$ are square windows that have the same
southwest corner and whose respective sizes are $cg^{j-2}$ and
$cg^{i-2}$, $W$ encloses $w$. The eastern side of $w + (x,0)$ still
lies within $W$, because the maximum value of $x$ is equal to the
difference between the size of $W$ and the size of $w$. The same
reasoning applies to a northward translation of $w$ by $(0,y)$. In
conclusion, $w+(x,y)$ must be enclosed in $W$, as long as neither $x$
nor $y$ exceeds $m$.
\end{proof}

Finally, in our main result, we will use the fact that, for any scaled
pier fractal $\mathbf{P}^c $, we can find an infinite number
of closed windows that all cut the fractal along a single line of
glues (see Lemma~\ref{lem:piers} below), the proof of which uses the
following three intermediate lemmas.

\begin{lemma} \label{lem:piers1}
If $\mathbf{P}$ is any pier fractal with generator $G$,
then $G$ contains at least one pier that is not a double-bridge pier.
\end{lemma}

\begin{proof}
For the sake of obtaining a contradiction, assume that $G$ contains
exactly one pier, say $(p,q)$, and that $(p,q)$ is a double-bridge
pier.  Note that any double-bridge pier must be positioned at one of
the corners of $G$, that is, $(p,q) \in
\{(0,0),(g-1,0),(0,g-1),(g-1,g-1)\}$.  Without loss of generality,
assume that $(p,q) = (g-1,0)$, as in Figure~\ref{fig:piers}
above. Since $(p,q)$ is a double-bridge pier, $(0,0)$ must be the
other point in the h-bridge and $(g-1,g-1)$ must be the other point in
the v-bridge. Thus, $(0,0)\in G$ (this is also true by definition of
$G$) and $(g-1,g-1)\in G$. Since $(p,q)$ is the only pier in $G$,
$(0,0)$ cannot be north-free (nor east-free), which implies that
$(0,1)\in G$. Therefore, $(g-1,1)\not\in G$ (otherwise, $G$ would
contain a second h-bridge). Similarly, since $(p,q)$ is the only pier
in $G$, $(g-1,g-1)$ cannot be west-free (nor south-free), which
implies that $(g-2,g-1)\in G$. Therefore, $(g-2,0)\not\in G$
(otherwise, $G$ would contain a second v-bridge). In conclusion, the
point $(p,q)=(g-1,0)$ is in $G$ but it is not connected to the rest of
$G$, which contradicts the definition of $\mathbf{P}$, whose generator
must be connected.  
\end{proof}

\begin{lemma} \label{lem:piers2}
Let $\mathbf{P}$ be any pier fractal with generator $G$
such that $(p,q)\in G$ is a parallel single-bridge pier. If $c \in
\mathbb{Z}^+$, then it is always possible to pick one point $(e,f)$ in
$G$ such that, for $1<s \in \N$, $W^c_s(e,f,p,q)$
encloses a configuration that is connected to $\mathbf{P}^c$ via a
single connected line of glues of length $c$.
\end{lemma}

\begin{proof}
Without loss of generality, assume that the pier $(p,q)$ is pointing
north, that it belongs to a v-bridge and that $q=g-1$ (a similar
reasoning holds if $q=0$ and the pier points south, or if the pier
belongs to an h-bridge and points either west or east). Now, we must
pick a point $(e,f)$ such that any window of the form $W^c_s(e,f,p,q)$
has exactly three free sides. We distinguish two cases.
\begin{enumerate}
\item If $p=0$, that is, the pier is in the leftmost column of $G$,
  then $(1,g-1) \not\in G$, since $(0,g-1)$ is a north-pointing
  pier. Therefore, there must exist at least one point in $G \cap
  (\{1\} \times \N_{g-1})$, say $(1,y)$, with $0 \leq y < g-1$, that
  is north-free. In this case, we pick $(e,f)$ to be equal to $(1,y)$.
  Now, consider any window $w$ of the form $W^c_s(e,f,p,q)$. The north
  side of $w$ is free (since $q=g-1$, $(e,f)$ is north-free in $G$
  and $f=y<g-1$), the east side of $w$ is free (since $(1,g-1)\not\in
  G$), and the west side of $w$ is free (since the facts that
  $(0,0)\in G$, $(0,g-1)\in G$ and $(0,g-1)$ is a single-bridge pier
  together imply that $(g-1,g-1) \not\in G$). Furthermore, since
  $(0,g-1)$ is a north-pointing pier, $S(0,g-1) \in G$.
\item If $p>0$, then $(p-1,g-1)\not\in G$. Therefore, there must exist
  at least one point in $G \cap (\{p-1\} \times \N_{g-1})$, say
  $(p-1,y)$, with $0 \leq y < g-1$, that is north-free. In this case,
  we pick $(e,f)$ to be equal to $(p-1,y)$.  Now, consider any window
  $w$ of the form $W^c_s(e,f,p,q)$. The north side of $w$ is free
  (since $q=g-1$, $(e,f)$ is north-free in $G$ and $f=y<g-1$), the west side of
  $w$ is free (because $(p-1,g-1)\not\in G$), and the east side of $w$
  is free (since, either $p<g-1$ and $(p+1,g-1)\not\in G$, or $p=g-1$,
  in which case the facts that $(g-1,g-1)\in G$, $(g-1,0)\in G$ and
  $(g-1,g-1)$ is a single-bridge pier together imply that $(0,g-1)
  \not\in G$). Furthermore, since $(p,g-1)$ is a north-pointing pier,
  $S(p,g-1) \in G$.
\end{enumerate}
Therefore, in both cases, $W^c_s(e,f,p,q)$ has exactly three free
sides and encloses a configuration that is connected
to $\mathbf{P}^c$ via a single connected horizontal line
of glues of length $c$ positioned on the south side of the window.

\end{proof}

\begin{lemma} \label{lem:piers3}
Let $\mathbf{P}$ be any pier fractal with generator $G$
such that $(p,q)\in G$ is an orthogonal  single-bridge pier. If $c \in
\mathbb{Z}^+$, then it is always possible to pick one point $(e,f)$ in
$G$ such that, for $1<s \in \N$, $W^c_s(e,f,p,q)$
encloses a configuration that is connected to $\mathbf{P}^c$ via a
single connected line of glues of length $c$.
\end{lemma}

\begin{proof}
Without loss of generality,
assume that the pier $(p,q)$ is pointing east, that it belongs to a v-bridge
and that $q=g-1$ (a similar reasoning holds if $q=0$, or if the pier
points west, or if the pier belongs to an h-bridge and points either
north or south). Note that, in this case, $g$ must be strictly greater
than 2, since $(p,g-1)\in G$, $(p,0)\in G$ but $(p,g-2)\not\in
G$. Now, we must pick a point $(e,f)$ such that any window of the form
$W^c_s(e,f,p,q)$ has exactly three free sides. We distinguish two
cases.
\begin{enumerate}
\item If $p<g-1$, then $(p,g-2)\not\in G$, since $(p,g-1)$ is an
  east-pointing pier. Therefore, there must exist at least one point
  in $G \cap (\{p\} \times \N_{g-2})$, say $(p,y)$, with $0 \leq y <
  g-2$, that is north-free. In this case, we pick $(e,f)$ to be equal
  to $(p,y)$.  Now, consider any window $w$ of the form
  $W^c_s(e,f,p,q)$. The north side of $w$ is free (since $q=g-1$, $(e,f)$ is
  north-free in $G$ and $f=y<g-2<g-1$), the east side of $w$ is free
  (since $p<g-1$ and $(p+1,g-1)\not\in G$), and the south side of $w$
  is free (since $(p,g-2)\not\in G$). Furthermore, since $(p,g-1)$ is
  an east-pointing pier, $W(p,g-1) \in G$.
\item If $p=g-1$, that is, the pier is in the rightmost column of $G$, then
the facts that $(g-1,0) \in G$, $(g-1,g-1)\in G$ and $(g-1,g-1)$ is a
single-bridge pier together imply that $(0,g-1)\not\in G$. This,
together with the fact that $(0,0)\in G$, implies that there must
exist at least one point in $G \cap (\{0\} \times \N_{g-1})$, say
$(0,y)$, with $0 \leq y < g-1$, that is north-free. In this case, we
pick $(e,f)$ to be equal to $(0,y)$. Now, consider any window $w$ of
the form $W^c_s(e,f,p,q)$. The north side of $w$ is free (since $q=g-1$,
$(e,f)$ is north-free in $G$ and $f=y<g-1$), the east side of $w$ is free
(since $(0,g-1)\not\in G$), and the south side of $w$ is free (since
$(p,g-2)\not\in G$). Furthermore, since
$(g-1,g-1)$ is an east-pointing pier, $W(g-1,g-1) \in G$.
\end{enumerate}
Therefore, in both cases, $W^c_s(e,f,p,q)$ has exactly three free
sides and encloses a configuration that is connected
to $\mathbf{P}^c$ via a single connected horizontal line
of glues of length $c$ positioned on the west side of the window.

\end{proof}

\begin{lemma} \label{lem:piers}
Let $\mathbf{P}$ be any pier fractal with generator $G$. If $c \in
\mathbb{Z}^+$, then it is always possible to pick one pier $(p,q)$
and one point $(e,f)$, both in $G$, such that, for $1<s \in \N$, $W^c_s(e,f,p,q)$ encloses a configuration that is connected
to $\mathbf{P}^c$ via a single connected (horizontal or vertical) line
of glues of length $c$.
\end{lemma}

\begin{proof}
Let $\mathbf{P}$ be any pier fractal with generator
$G$. Let $c \in \mathbb{Z}^+$ and $1<s \in \N$. By
definition of a pier  fractal, $G$ contains at least one
pier. We will pick one of these piers carefully.

According to Lemma~\ref{lem:piers1}, it is always possible to choose a
pier in $G$ that is not a double-bridge pier. Therefore, we can always
choose either a real pier or a single-bridge pier. We now consider
the three possible cases.

First, if $G$ contains one or more real piers, we can
simply choose one of them as $(p,q)$. In this case, we pick
$(e,f)=(p,q)$, since any window of the form $W^c_s(p,q,p,q)$, where
$(p,q)$ is a real pier in $G$, must have exactly three free
sides. Therefore, $W^c_s(p,q,p,q)$ must enclose a configuration that
is connected to $\mathbf{P}^c$ via a single line of glues of length
$c$, namely on its non-free side.

Second, if $G$ does not contain any real piers, $G$ must contain at least
one single-bridge pier. So we wrap up this proof by considering the
two types of single-bridge piers.

If $G$ contains at least one parallel single-bridge pier, according to
Lemma~\ref{lem:piers2}, it is always possible to choose one pier
$(p,q)$ and one point $(e,f)$, both in $G$, such that, for $1<s \in
\N$, $W^c_s(e,f,p,q)$ encloses a configuration that is connected to
$\mathbf{P}^c$ via a single connected line of glues of length $c$.

Finally, if $G$ contains at least one orthogonal single-bridge pier,
according to Lemma~\ref{lem:piers3}, it is always possible to choose
one pier $(p,q)$ and one point $(e,f)$, both in $G$, such that, for
$1< s \in \N$, $W^c_s(e,f,p,q)$ encloses a configuration that is
connected to $\mathbf{P}^c$ via a single connected line of glues of
length $c$.  
\end{proof}

\subsection{Main result}\label{sec:main-result}

We are now ready to prove our main result.

\begin{theorem}\label{thm:main}
Let $\mathbf{P}$ be any pier fractal. If $c \in
\Z^+$, then $\mathbf{P}^c$ does not strictly self-assemble in
the aTAM.
\end{theorem}

\begin{proof}

Let $\mathbf{P}$ be any pier fractal with a $g\times g$
generator $G$, where $1<g \in \N$. Let $c$ be any
positive integer.  For the sake of obtaining a contradiction, assume
that $\mathbf{P}^c$ does strictly self-assemble in some TAS
$\mathcal{T} = (T,\sigma,\tau)$. Further assume that $\vec{\alpha}$ is
some assembly sequence in $\mathcal{T}$ whose result is $\alpha$, such
that $\dom \alpha = \mathbf{P}^c$.

\begin{figure}[htp]
\centering
 \includegraphics[width=0.7\textwidth]{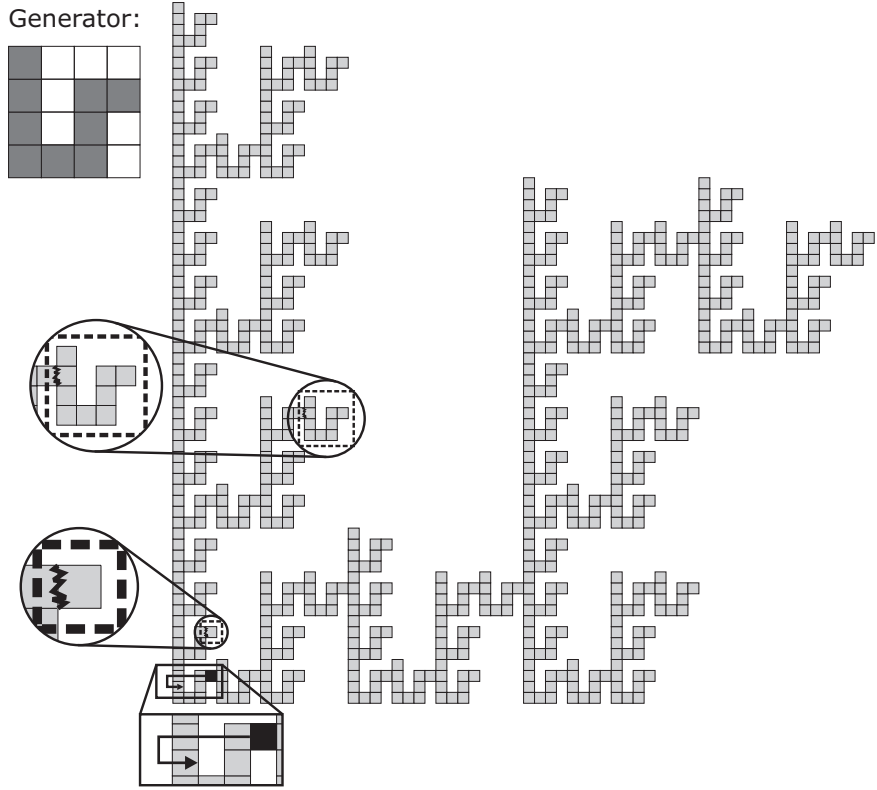}
\caption{First three stages ($s=1,2,3$) of an unscaled ($c=1$) pier fractal with an
  east-pointing pier at position $(3,2)$. The
  east-free point $(0,1)$ is at the tip of the arrow (see the rectangular magnification box). In other
  words,  $g=4$, $(p,q)=(3,2)$, and $(e,f)=(0,1)$.}
\label{fig:case2example}
\end{figure}

\begin{figure}
\centering

\renewcommand{\arraystretch}{1.5}
\begin{tabular}{|c|ll|}\cline{2-3}
\multicolumn{1}{c|}{} & \multicolumn{2}{c|}{\bf Translation Formulas for $(x,y)$}\\\hline


\raisebox{0.5\totalheight}{
  \begin{tabular}{c}
    {\bf North-pointing Pier} \\ \\
    Vertical Bridge\\
    with Southern \\
    Bridge Point at $(a,0)$\\
  \end{tabular}
}
&
      \includegraphics[height=1.2in]{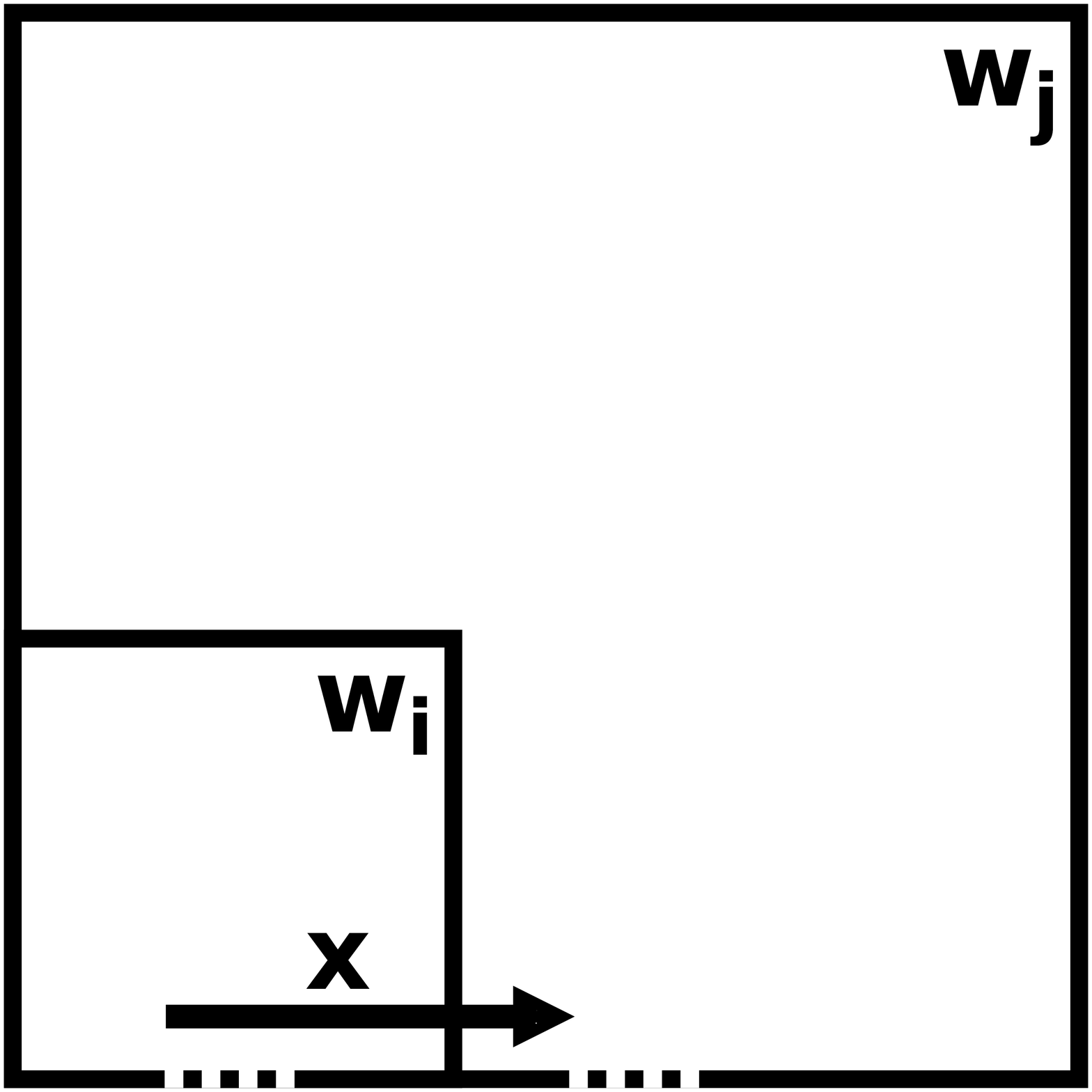}
&
\raisebox{1\totalheight}{
   \begin{tabular}{l}
             $x = ac \Sigma_{k=i-2}^{j-3}g^k$
        \\ \\
         $y = 0$        \\
     \end{tabular}
}\\\hline


\raisebox{0.5\totalheight}{
  \begin{tabular}{c}
    {\bf East-pointing Pier} \\ \\
    Horizontal Bridge\\
    with Western \\
    Bridge Point at $(0,b)$\\
  \end{tabular}
}
&
      \includegraphics[height=1.2in]{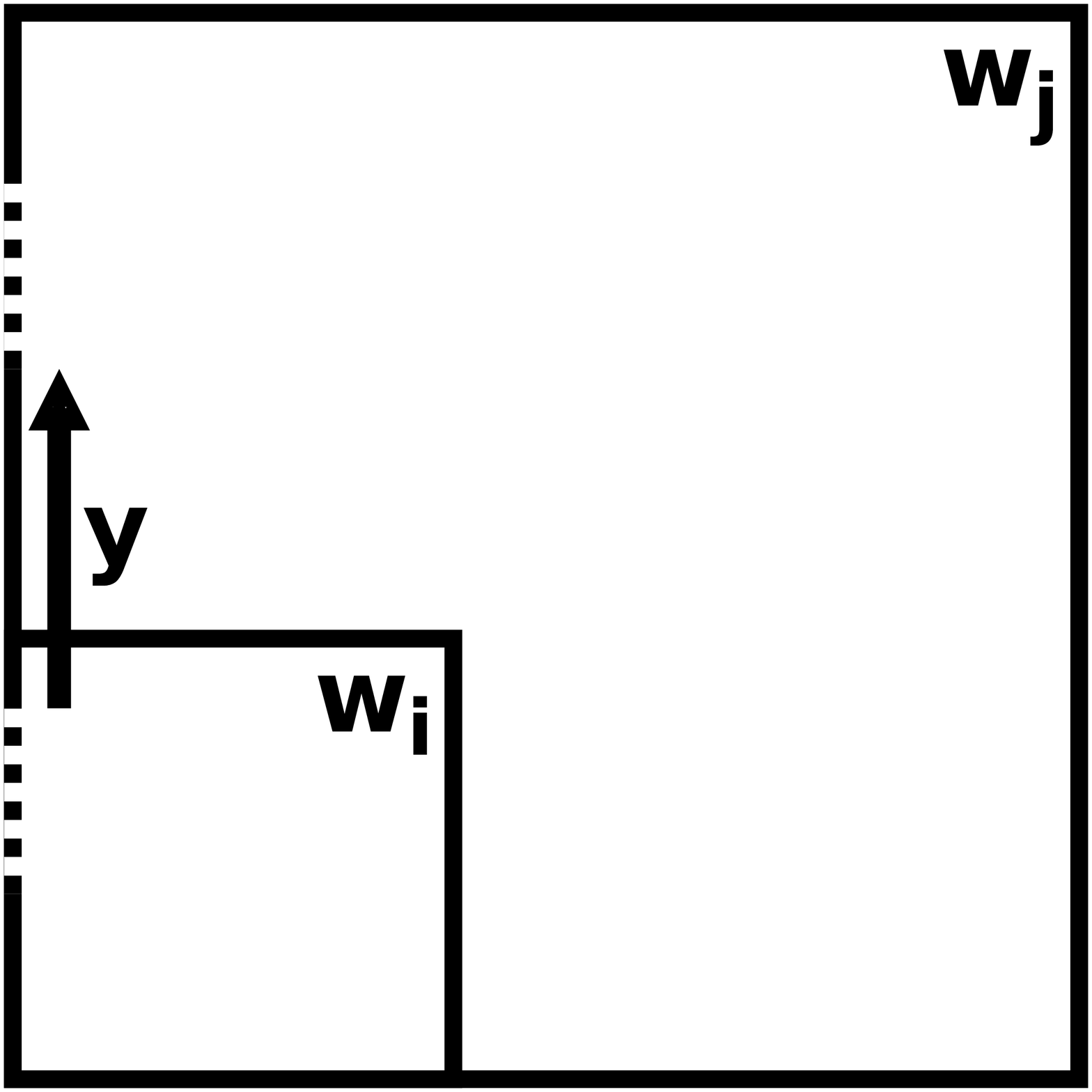}
&
\raisebox{1\totalheight}{
   \begin{tabular}{l}
             $x = 0$
        \\ \\
         $y= bc\Sigma_{k=i-2}^{j-3}g^k$        \\
     \end{tabular}
}\\\hline


\raisebox{0.5\totalheight}{
  \begin{tabular}{c}
    {\bf South-pointing Pier} \\ \\
    Vertical Bridge\\
    with Northern \\
    Bridge Point at $(a,g-1)$\\
  \end{tabular}
}
&
      \includegraphics[height=1.25in]{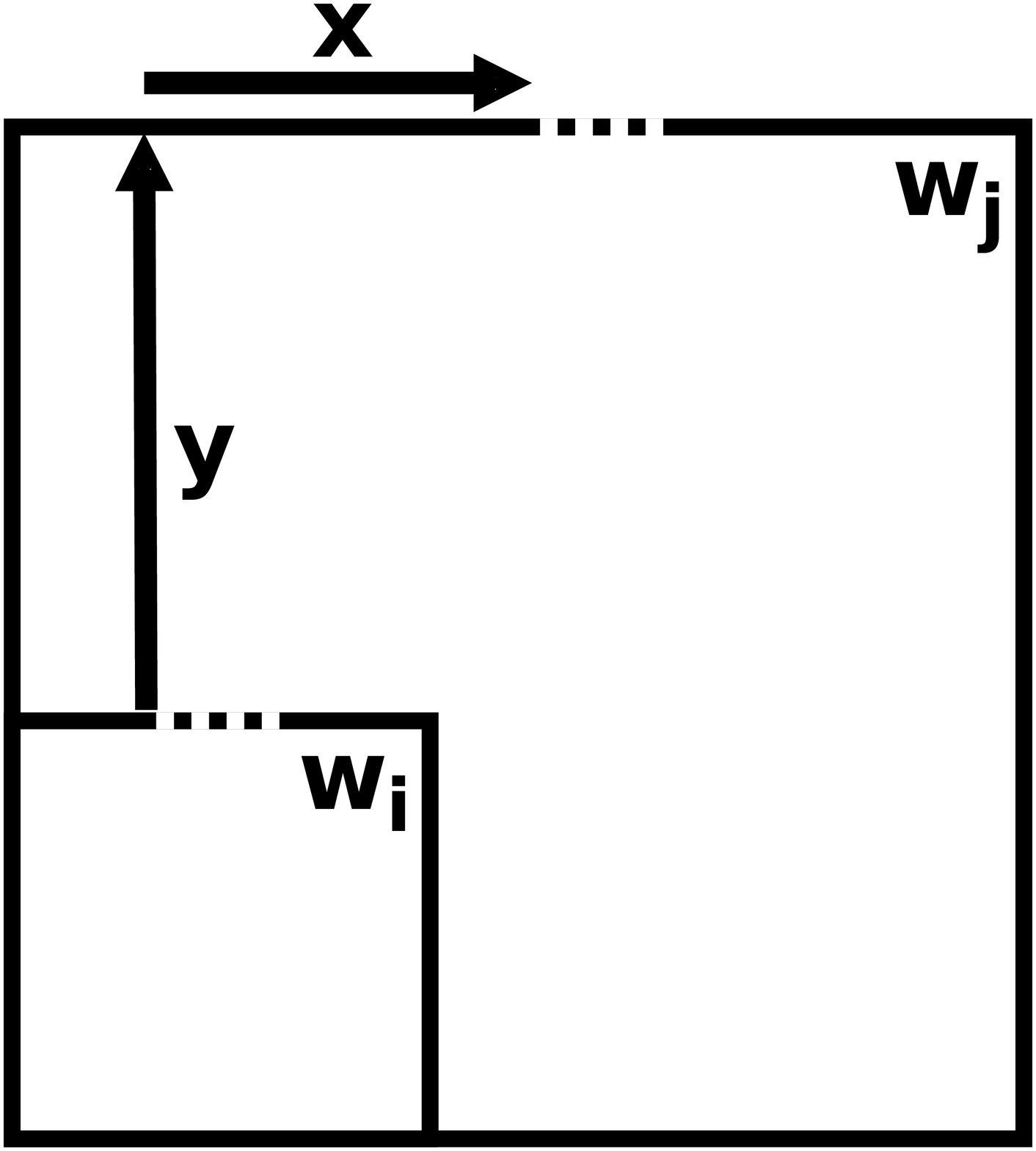}
&
\raisebox{1\totalheight}{
   \begin{tabular}{l}
             $x = ac\Sigma_{k=i-2}^{j-3}g^k$
        \\ \\
         $ y=c(g^{j-2}-g^{i-2})$
     \end{tabular}
}\\\hline


\raisebox{0.5\totalheight}{
  \begin{tabular}{c}
    {\bf West-pointing Pier} \\ \\
    Horizontal Bridge\\
    with Eastern \\
    Bridge Point at $(g-1,b)$\\
  \end{tabular}
}
&
      \includegraphics[height=1.2in]{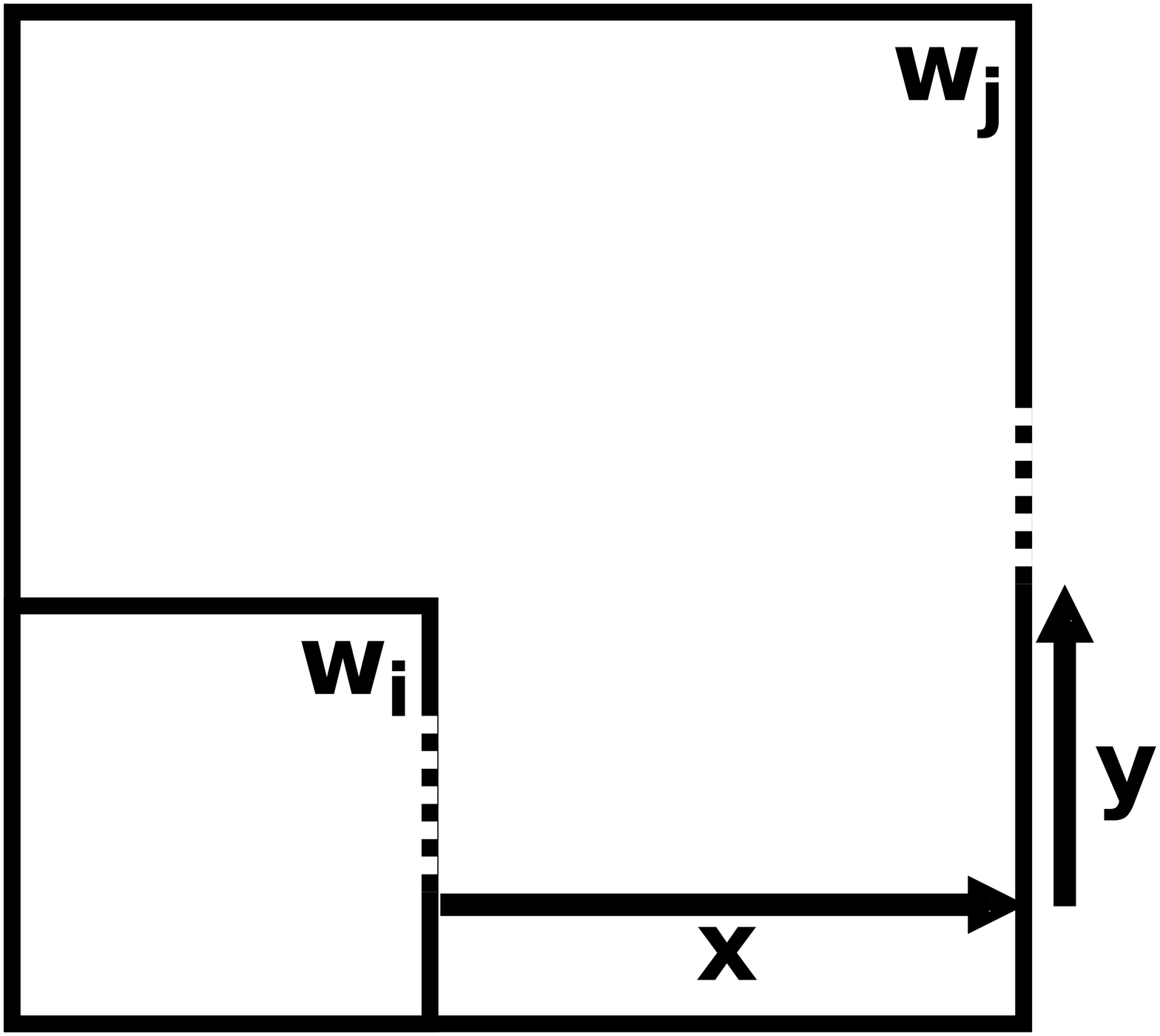}
&
\raisebox{1\totalheight}{
   \begin{tabular}{l}
             $x =c(g^{j-2}-g^{i-2})$
        \\ \\
         $ y= bc\Sigma_{k=i-2}^{j-3}g^k$
     \end{tabular}
}\\\hline
\end{tabular}
\caption{Computing the coordinates $(x,y)$ of the translation that
  aligns the bond-forming glues (shown as a dotted line) of the
  windows $w_i$ and $w_j$. Note that $(a,b)$ with $a\in\N_g$ and $b\in\N_g$
  are the coordinates of the point in $G$ within the (horizontal or
  vertical) bridge that determines the position of the bond-forming glues.}
\label{fig:translationcases}
\end{figure}

\begin{figure}[htp]
\centering
 \includegraphics[width=0.95\textwidth]{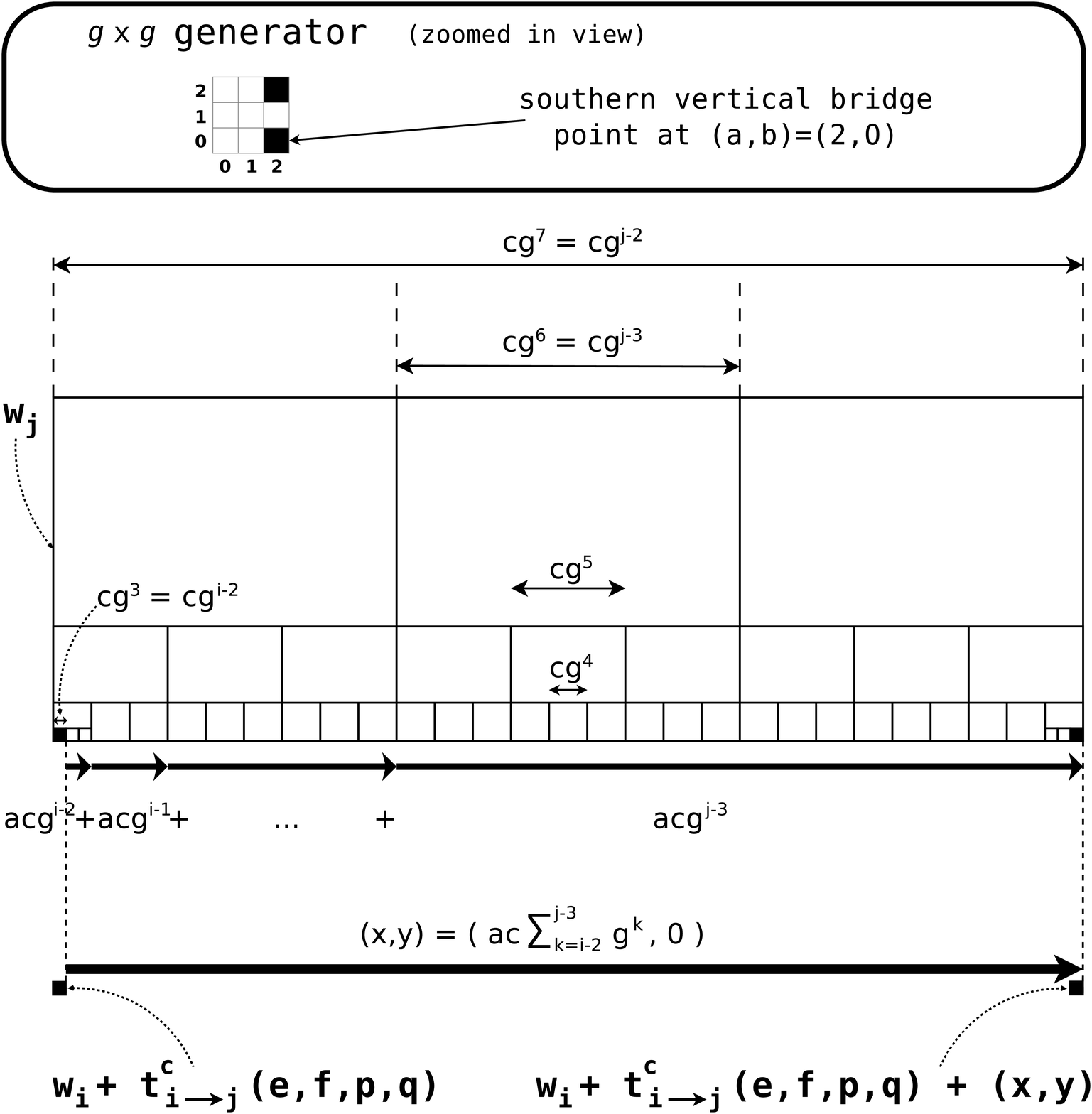}
\caption{$(x,y)$ translation needed to align $w_i$ and $w_j$ on their
  east side once their southwest corners already match. Example with a
  north-pointing pier (not shown) and $g=3$, $i=5$, $j=9$,
  and southern vertical bridge point at location $(a,b)=(2,0)$.}
\label{fig:translation}
\vspace{-10pt}
\end{figure}

According to Lemma~\ref{lem:piers}, we can always pick one pier
$(p,q)$ and a point $(e,f)$, both in $G$, such that, for $1<s \in \N$,
the window $W^c_s(e,f,p,q)$, which we will abbreviate $w_s$, encloses
a configuration that is connected to $\mathbf{P}^c$ via a single line
of glues of length $c$.\footnote{Without loss of generality, we will
  assume that this line of glues is positioned on the western side of
  the windows and is thus vertical (see the jagged lines in
  Figure~\ref{fig:case2example}, where $s=2$ and $s=3$ for the small
  and large windows, respectively, and $(p,q) = (3,2)$ and
  $(e,f)=(0,1)$), because the chosen pier in our example points
  east. A similar reasoning holds for piers pointing north, south or
  west.}  The maximum number of distinct combinations and orderings of
glue positionings along this line of glues is
finite.\footnote{\label{note:combinatorics}This number is bounded
  above by $T_{glue}^{2c}\cdot(2c)!$, where $T_{glue}$ is the total
  number of distinct glue types in $T$.}  By the generalized
pigeonhole principle, since
$\left|\left\{w_s~\left|~1<s\in\N\right.\right\}\right|$ is infinite,
there must be at least one bond-forming submovie such that an infinite
number of windows generate this submovie (up to translation). Let us
pick two such windows, say, $w_i$ and $w_j$ with $i<j$, such that
$\mathcal{B}(M_{\vec{\alpha},w_i})$ and
$\mathcal{B}(M_{\vec{\alpha},w_j})$ are equal (up to translation). We
must pick these windows carefully, since as stated in the proof of
Lemma~\ref{lem:wml}, the seed of $\alpha$ must be either in both
windows or in neither. This condition can always be satisfied. The
only case where the seed is in more than one window is when it is at
position $(0,0)$ and $e=f=p=q=0$, which implies that all windows
include the origin. So, in this case, any choice of $i$ and $j>i$ will
do. In all other cases, none of the windows overlap. So, if the seed
belongs to one of them, say $w_k$, then we can pick any $i$ greater
than $k$ (and $j>i$). Finally, if the seed does not belong to any
windows, then any choice of $i$ and $j>i$ will do.

We will now prove that $w_i$ and $w_j$ satisfy the two
conditions of Lemma~\ref{lem:wml}.

First, we compute $\vec{c}$ such that
$\mathcal{B}(M_{\vec{\alpha},w_i}) +\vec{c} =
\mathcal{B}(M_{\vec{\alpha},w_j})$. We know that $w_i
+\vec{t}^c_{i\rightarrow j}(e,f,p,q)$ and $w_j$ share their southwest
corner.  We need to perform one more translation $(x,y)$ to align the
bond-forming glues of $w_i$ and $w_j$. The values of $x$ and $y$
depend on the direction in which the chosen pier is pointing. The
formulas corresponding to all four directions are given in
Figure~\ref{fig:translationcases}. Furthermore, a justification for
the recurring summation in the formulas of
Figure~\ref{fig:translationcases} is provided in
Figure~\ref{fig:translation}. In that figure, the chosen case is a
north-pointing pier. However, our example in
Figure~\ref{fig:case2example} above uses an east-pointing pier. We now
complete the proof for this case.
To align the bond-forming glues of $w_i$ and $w_j$, we must translate
$w_i +\vec{t}^c_{i\rightarrow j}(e,f,p,q)$ by $(x,y) =
\left(0,bc\sum_{k=i-2}^{j-3}g^k\right)$, with $b\leq g-1$.  Since $x=0\leq m$ (as
defined in Lemma~\ref{lem:enclosure2}) and $y=bc\sum_{k=i-2}^{j-3}g^k
\leq (g-1)c\sum_{k=i-2}^{j-3}g^k =
c\left(\sum_{k=i-1}^{j-2}g^k-\sum_{k=i-2}^{j-3}g^k\right)
=c\left(g^{j-2}-g^{i-2}\right)=m$, we can apply
Lemma~\ref{lem:enclosure2} to infer that, with $\vec{c} =
\vec{t}^c_{i\rightarrow j}(e,f,p,q) + (x,y)$, $w_i+\vec{c}$ is
enclosed in $w_j$. Therefore, the second condition of
Lemma~\ref{lem:wml} holds.

Second, by construction, $\mathcal{B}(M_{\vec{\alpha},w_i}) +\vec{c} =
\mathcal{B}(M_{\vec{\alpha},w_j})$.  Therefore, the first condition of
Lemma~\ref{lem:wml} holds.

In conclusion, the two
conditions of Lemma~\ref{lem:wml} are satisfied, with $\alpha_I$ and
$\alpha_O'$ defined as the intersection of $\mathbf{P}^c$ with the inside
of $w_i$ and the outside of $w_j$, respectively. We can thus conclude
that the assembly $\alpha_O' \cup (\alpha_I +\vec{c})$ is producible
in $\mathcal{T}$. Note that this assembly is identical (up to
translation) to $\mathbf{P}^c$, except that the interior of the large
window $w_j$ is replaced by the interior of the small window
$w_i$. Since the configurations in these two windows cannot be
identical, we have proved that $\mathcal{T}$ does not strictly
self-assemble $\mathbf{P}^c$, which is a contradiction.

\end{proof}

\subsection{Corollaries of our main result}

In this section, we discuss both special cases and generalizations of
our main result.

\subsubsection{Specializations of our main result}

In~\cite{ScaledTreeFractals}, we proved that no scaled tree fractal strictly
self-assembles in the aTAM, where a \emph{tree fractal} is a discrete
self-similar fractal whose underlying graph is a tree. In this
section, we start by proving a new characterization of tree fractals
in terms of simple connectivity properties of their generator.

\begin{theorem}\label{thm:tree-fractal}
$\displaystyle\mathbf{T} = \bigcup_{i=1}^{\infty}{T_i}$ is a $g$-discrete self-similar tree fractal, for some $g > 1$, with generator $G$ if and only if \\
\hspace*{1cm}a. $G$ is a tree and\\
\hspace*{1cm}b. $nhb_G = nvb_G =1.$
\end{theorem}

The proof of this theorem is in the appendix. Next, the following
observation follows from the fact that a tree with more than one
vertex must contain at least two leaf nodes.

\begin{observation}
\label{obs:piers}
If $G$ is the generator of any discrete self-similar fractal and $G$
is a tree, then it must contain at least two piers.
\end{observation}

Finally, we can recast the main result
in~\cite{ScaledTreeFractals} as a special case of our
main result.

\begin{corollary}\label{cor:tree}[From \cite{ScaledTreeFractals}]
Let $\mathbf{T}$ be any tree fractal. If $c \in
\Z^+$, then $\mathbf{T}^c$ does not strictly self-assemble in
the aTAM.
\end{corollary}

\begin{proof}
Let $\mathbf{T}$ be any tree fractal with generator $G$.  According to
Theorem~\ref{thm:tree-fractal}, the full grid graph of $G$
is a tree and is thus connected, and $nhb_G = nvb_G
=1$. Furthermore, according to Observation~\ref{obs:piers}, $G$ must
contain at least one pier. Therefore, $\mathbf{T}$ is a
pier fractal and $\mathbf{T}^c$ does not strictly
self-assemble in the aTAM.  
\end{proof}

We now turn our attention to a second specialization of our main result by
considering ``pinch-point fractals,'' which are defined in
\cite{jSADSSF} as follows.

\begin{definition} \label{def:pinch-point}
Let $\mathbf{X} \subset \N^2$ be a $g$-discrete self-similar fractal
with generator $G$. We say that $\mathbf{X}$ is a \emph{pinch-point
  discrete self-similar fractal} if $G$ satisfies the following four
conditions:
\begin{enumerate}
\item $\{(0, 0), (0, g - 1), (g - 1, 0)\} \subseteq G$.
\item $G \cap (\{1,\ldots,g-1\} \times \{g -1\}) = \emptyset$.
\item $G \cap (\{g -1\} \times \{1,\ldots,g-1\} ) = \emptyset$.
\item The full grid graph of  $G$ is connected.
\end{enumerate}
\end{definition}

Theorem~3.12 in \cite{jSADSSF} establishes that no pinch-point fractal
strictly self-assembles in the aTAM. We can now generalize this result
as follows.

\begin{corollary}\label{cor:pinch_point}
Let $\mathbf{X}$ be any pinch-point discrete self-similar fractal. If
$c \in \Z^+$, then $\mathbf{X}^c$ does not strictly self-assemble in
the aTAM.
\end{corollary}

\begin{proof}
Let $\mathbf{X}$ be any pinch-point discrete self-similar fractal with
generator $G\subset \N_g^2$, for some $g>1$. First, by definition of a
pinch-point fractal, the full grid graph of $G$ is connected. Second,
since the point $(g-1,0)$ is the only point of $G$ that belongs to
$\{g-1\}\times \N$ and the point $(0,0)$ also belongs to $G$, $nhb_G =
1$. Similarly, since the point $(0,g-1)$ is the only point of $G$ that
belongs to $\N \times \{g-1\}$ and the point $(0,0)$ also belongs to
$G$, $nvb_G = 1$. Third, since the point $(0,g-1)$ is a pier in $G$,
$G$ contains at least one pier. Therefore, $\mathbf{X}$ is a
pier fractal. In conclusion, if $c \in \Z^+$, then
$\mathbf{X}^c$ does not strictly self-assemble in the aTAM.
  
\end{proof}

\subsubsection{Generalizations of our main result}\label{sec:generalizations}

We now discuss how to extend our main result to different classes of
fractals. More specifically, we will relax the last two conditions in
the definition of pier  fractals and still be able to use
the same reasoning as we did in the proof of our main result.

\begin{figure}[htp]
\centering
\includegraphics[width=0.9\textwidth]{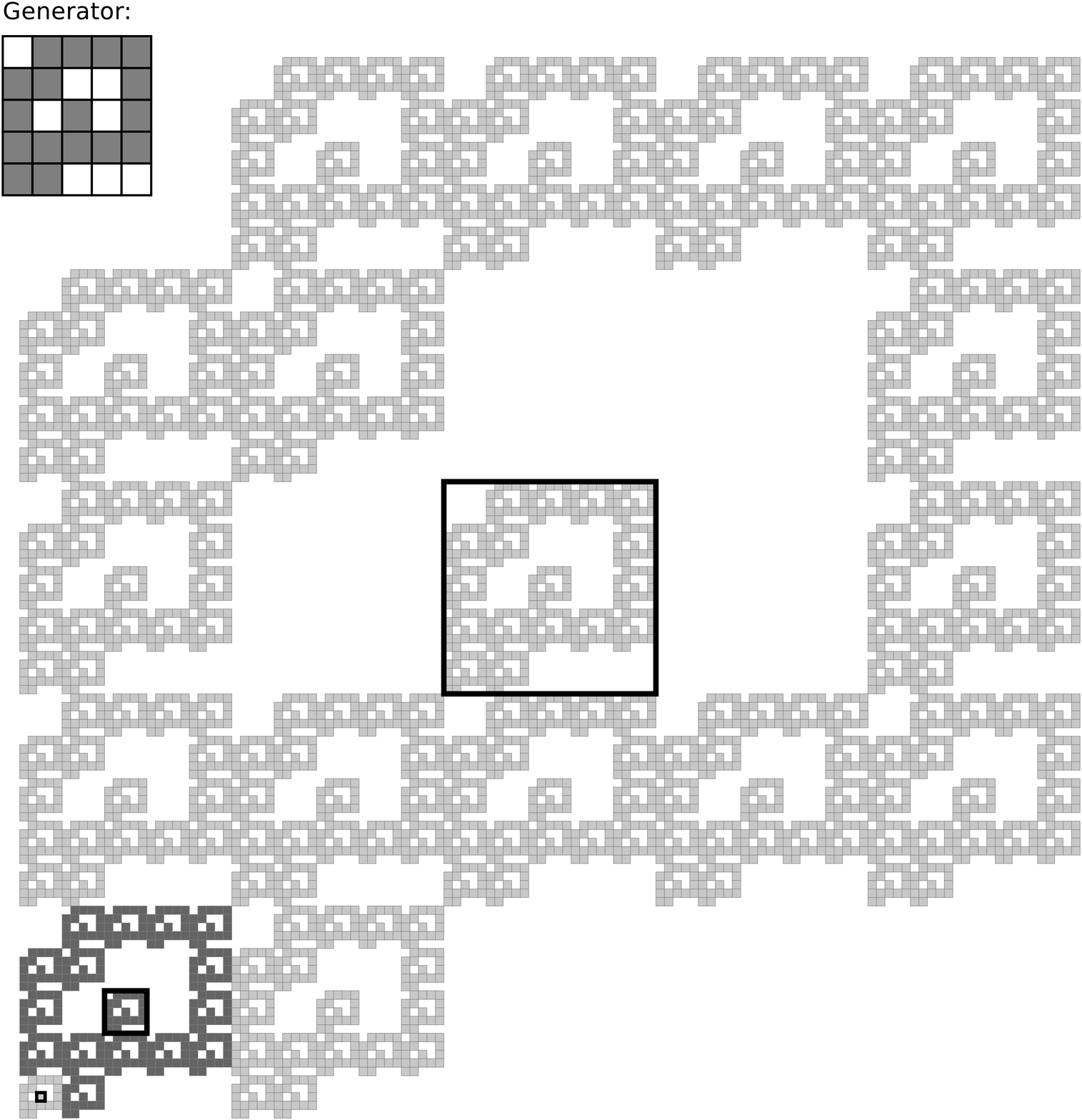}
\caption{First three stages ($s=1,2,3$) of an unscaled ($c=1$) 5-discrete
  self-similar fractal with a north-pointing pier,
  $nhb_G=3$, $nvb_G=1$, $(p,q)=(2,2)$, and $(e,f)=(0,0)$.}
\label{fig:corollary3}
\end{figure}

First, our proof of Theorem~\ref{thm:main} uses the fact that there
exist an infinite collection of square windows, each of which encloses
a sub-configuration of the fractal that is attached to the rest of the
fractal at a single point (or single line of points). In other words,
each window in the collection has three free sides. If, for example,
the east and west sides of each window are free, then the number of
horizontal bridges in the generator $G$ does not matter. Even if
$nhb_G>1$, our construction for the windows still
works. Figure~\ref{fig:corollary3} is one example of such a fractal to
which our main result generalizes, with the first three windows shown
as thick, black squares. In this case, our proof technique still
works, even though the generator contains three horizontal bridges.
Here is a precise statement of the corollary.

\begin{corollary}\label{cor:multiple_bridges}
Let $\mathbf{F}$ be a discrete self-similar fractal with generator $G$
such that the full grid graph of $G$ is connected, $nvb_G = 1$, and
$G$ contains at least one north-pointing pier or one south-pointing
pier. If $c \in \Z^+$, then $\mathbf{F}^c$ does not strictly
self-assemble in the aTAM.
\end{corollary}

Symmetrically, a similar result holds for fractals whose generator $G$
contains either at least one west-pointing pier or at least one
east-pointing pier, and such that $nhb_G=1$.

Second, having relaxed the second condition (part $b$) of the
definition of pier fractals, we can now relax the third
condition (part $c$) as well. To apply our Closed Window Movie Lemma,
a pier is not strictly needed. Instead, the generator only need
contain a \emph{pier-like sub-configuration}, that is, a sub-configuration of
one or more tiles that is attached to the rest of the fractal at a
single point.  Figure~\ref{fig:pier_like}
gives one example of such a fractal with the first two windows shown
as thick, solid, black squares. In this case, our proof technique still works, even though the
generator contains five horizontal bridges and no pier.
Here is a precise statement of the corollary.

\begin{figure}[htp]
\centering
 \includegraphics[width=0.8\textwidth]{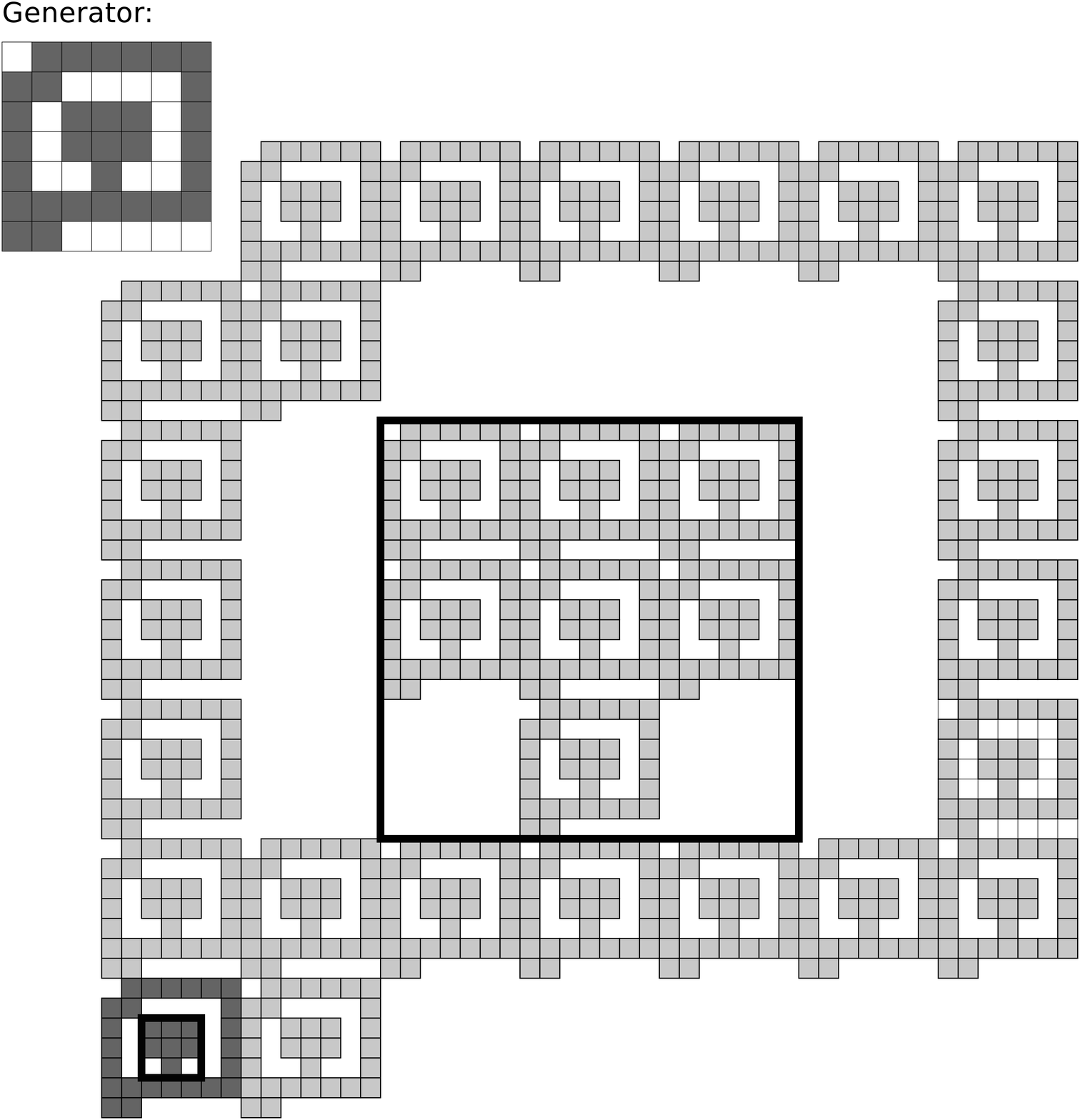}
\caption{First two stages ($s=1,2$) of an unscaled ($c=1$) 7-discrete
  self-similar fractal with a north-pointing pier-like sub-configuration,
  $nhb_G=5$ and $nvb_G=1$.}
\label{fig:pier_like}
\end{figure}

\begin{corollary}\label{cor:no_pier}
Let $\mathbf{F}$ be a discrete self-similar fractal with generator $G$
such that the full grid graph of $G$ is connected, $nvb_G = 1$, and
$G$ contains at least one north-pointing pier-like sub-configuration
or at least one south-pointing pier-like sub-configuration. If $c
\in \Z^+$, then $\mathbf{F}^c$ does not strictly self-assemble in the
aTAM.
\end{corollary}

Symmetrically, a similar result holds for fractals whose generator $G$
contains either at least one west-pointing pier-like
sub-configuration or at least one east-pointing pier-like
sub-configuration, and such that $nhb_G=1$.

Finally, the Closed Window Movie Lemma may be applicable even when the
generator does not contain any pier-like sub-configuration. The key
requirement in the proof of our main result is to be able to find at
least two windows that share a common bond-forming window movie but
whose insides contain different sub-configurations. This requirement
can be met even when the sub-configuration contained in each window is
attached to the rest of the fractal at more than one point. Figure~\ref{fig:same-column} illustrates such a
situation. For a general characterization, we need some
definitions.

If $G$ is a $g \times g$ generator, then a \emph{column} is any set $G
\cap (\{x\} \times \N_g$), where $x \in \N_g$ is the index of the
column. Therefore, columns are indexed from left to right starting at
0. Two columns are \emph{equivalent} if they contain the same number
of points and, for each point in one column, there is a point in
the other column with the same $y$ coordinate.  A \emph{vertical cut}
is any set of edges connecting two adjacent columns of $G$. Two
vertical cuts are \emph{equivalent} if they contain the same number of
edges and, for each edge in one cut, there is an edge in the other
cut with the same $y$ coordinate. Figure~\ref{fig:same-column} depicts
a $5\times 5$ generator $G$ in which columns 2 and 3 are
equivalent. Furthermore, in this example, cuts 2 and 3 are also
equivalent.\footnote{The index of a vertical cut is given by the index
  of the leftmost of the two columns that its edges connect.} Note
that the fact that two columns $i$ and $j$ are equivalent does not
imply that the vertical cuts $i$ and $j$ are also equivalent. That
both facts hold is just a coincidence in this example. If, for
example, the point $(4,1)$ were removed from the generator in
Figure~\ref{fig:same-column}, then columns 2 and 3 would still be
equivalent, but vertical cuts 2 and 3 would no longer be
equivalent.\footnote{We included point $(4,1)$ in the generator to
  exclude piers from the generator in this example.} In general, there is no
correlation between the indices of equivalent columns and the indices
of equivalent vertical cuts. However, the co-existence of equivalent
columns and equivalent vertical cuts in the same generator may render
the Window Movie Lemma applicable.

\begin{figure}[htp]
\centering
 \includegraphics[width=\textwidth]{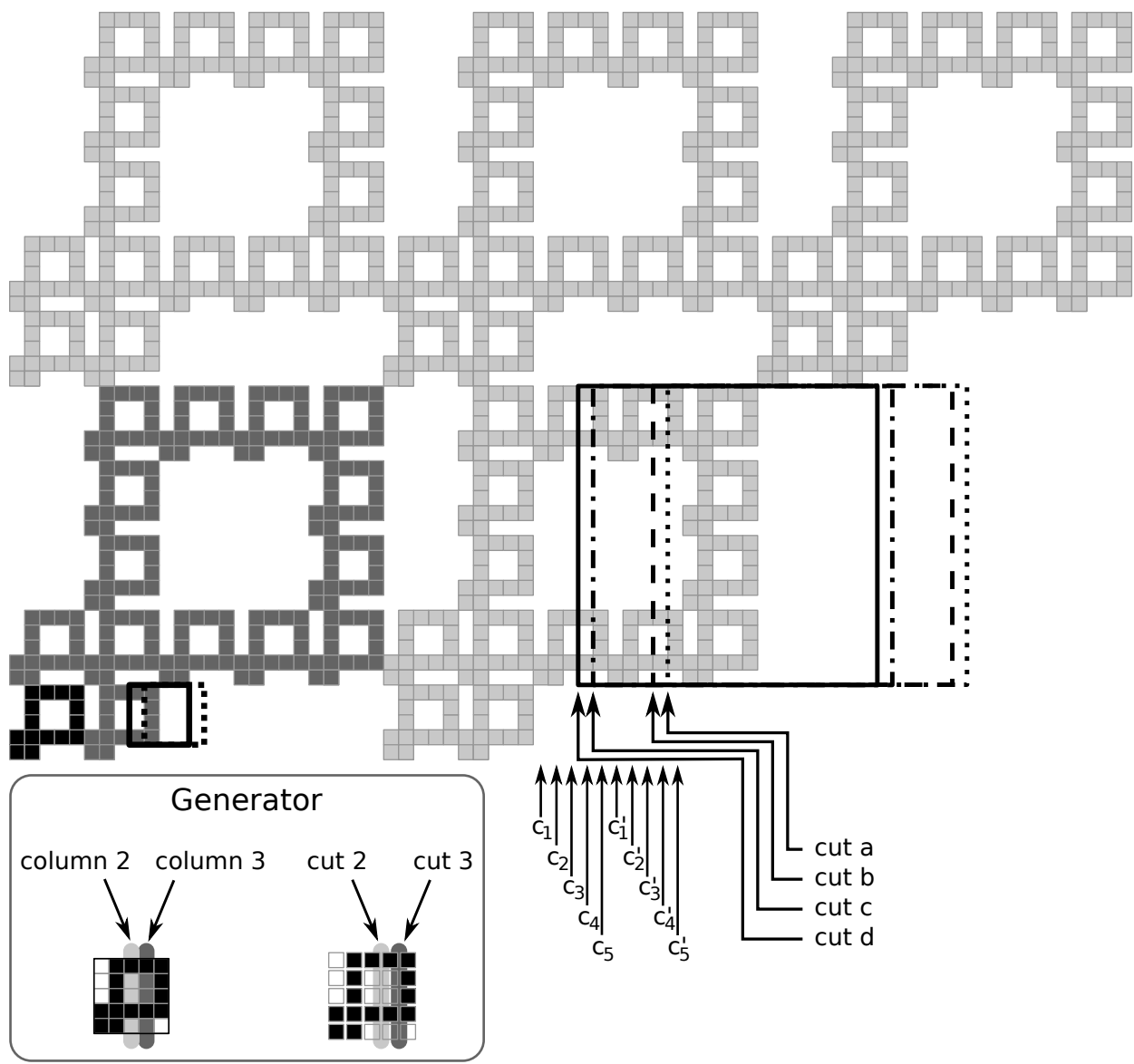}
\caption{First two stages (and part of the third stage) of an
unscaled ($c=1$) 5-discrete self-similar fractal with two equivalent
columns and two equivalent vertical cuts.}
\label{fig:same-column}
\end{figure}

In the example of Figure~\ref{fig:same-column}, vertical cut 2 is to
the east of the vertical bridge (which is a subset of column
1). Therefore, if we can find an east-free point in $G$, e.g., the
point $(1,0)$ in our running example, we will be able to position a
closed window that only cuts the fractal on one side (here, its
western side), e.g., the smallest of the two solid windows in
Figure~\ref{fig:same-column}. Similarly, we can position another
closed window of the same size that cuts the generator through
vertical cut 3, e.g., the dotted window that overlaps the small solid
window. By construction, the window movies corresponding to these two
windows have the same length and contain exactly the same positions
(up to translation). Of course, these window movies may not be equal
up to translation because the glues in their respective positions may
not match. But this is where we can take advantage of the existence of
two equivalent columns. By self-similarity, these two columns will, in
the next stage of the fractal, become two 5-wide sets of columns of
height 20 that are pairwise equivalent, that is, columns $c_1$ and
$c'_1$ are equivalent, columns $c_2$ and $c'_2$ are equivalent,
\ldots, and columns $c_5$ and $c'_5$ are equivalent. More importantly,
the 20-high cuts labeled a, b, c, and d in
Figure~\ref{fig:same-column} are all pairwise equivalent. Therefore,
at this stage of the fractal, we can build four larger square windows,
as shown in Figure~\ref{fig:same-column}. Furthermore, at each
successive stage of the fractal, we will be able to build
twicef\footnote{Note that, in this example, cut 1 is also equivalent
  to cuts 2 and 3. So we can actually build three windows for each
  one of the equivalent columns in the generator. Therefore, we could
  could have drawn 3, 6, 12, etc. windows for stages 2, 3, 4, etc.,
  respectively.  However, we chose to use only two of the three
  equivalent cuts in our discussion in order to keep the figure as
  legible as possible.}  as many square windows that all generate
window movies of the same length and with positions that are equal up
to translation. Since the number of window movies grows without bound
as the stage number increases, but the number of distinct combinations
and orderings of glue positionings is finite (following a reasoning
similar to the one in Footnote~\ref{note:combinatorics}), there is
always a stage (in fact, an infinite number of them) that contains two
bond-forming window movies that are identical up to translation. The
sub-configurations inside the two corresponding windows cannot be
equivalent because of the way the windows overlap. Additionally, since
the two windows have exactly the same shape and size, the translation
of the eastmost one is enclosed in (in fact, equal to) the other
one. Therefore, we can apply the Closed Window Movie Lemma and
conclude the proof by contradiction.  Here is a precise statement of
the corollary that covers the class of similar situations.

\begin{corollary}\label{cor:same_column}
Let $\mathbf{F}$ be a discrete self-similar fractal with generator $G$
such that the full grid graph of $G$ is connected, $G$
contains two equivalent columns, and $G$ contains two equivalent
vertical cuts that are positioned on the same side of all vertical
bridges. If $c \in \Z^+$, then $\mathbf{F}^c$ does not strictly
self-assemble in the aTAM.
\end{corollary}

Symmetrically, a similar result holds for fractals with equivalent
rows and equivalent horizontal cuts.

To conclude this section, we note that Corollary~\ref{cor:same_column}
could have been proved using the standard Window Movie Lemma
introduced in \cite{WindowMovieLemma}, since the windows used in the
proof have exactly the same shape and size. In the next section, we motivate
our introduction of the Closed Window Movie Lemma as a more
convenient tool in the study of scaled pier fractals.

\section{Discussion}

A fair question for one to ask is: why not simply prove
Theorem~\ref{thm:main} using the standard Window Movie Lemma from
\cite{WindowMovieLemma}? Our response is that we currently do not know
that we cannot.

For the sake of discussion, the statement of the standard WML, restricted to bond-forming submovies, is as follows.

\begin{lemma}[Standard Window Movie Lemma \cite{WindowMovieLemma}]\label{lem:swml}
Let $\vec{\alpha}~=~(\alpha_i~|~0\leq~i<~l)$ and
$\vec{\beta}~=~(\beta_i~|~0\leq~i<~m)$, with $l,m \in \mathbb{Z}^+
\cup \{\infty\}$, be assembly sequences in some TAS $\mathcal{T}$ with
results $\alpha$ and $\beta$, respectively. Let $w$ be a window that
partitions $\alpha$ into two configurations $\alpha_L$ and $\alpha_R$,
and, for some $\vec{c}\neq(0,0)$, $w' = w + \vec{c}$ be a translation
of $w$ that partitions $\beta$ into two configurations $\beta_L$ and
$\beta_R$.  Furthermore, define $M_{\vec{\alpha},w}$ and
$M_{\vec{\beta},w'}$ to be the respective window movies for
$\vec{\alpha},w$ and $\vec{\beta},w'$ and define $\alpha_L, \beta_L$
to be the sub-configurations of $\alpha$ and $\beta$ containing the
seed tiles of $\alpha$ and $\beta$, respectively. Then, if
$\mathcal{B}\left(M_{\vec{\alpha},w}\right) + \vec{c} =
\mathcal{B}\left(M_{\vec{\beta},w'}\right)$, it is the case that the
following two assemblies are also producible: (1) the assembly
$\alpha_L\beta'_R = \alpha_L \cup \beta'_R$ and (2) the assembly
$\beta'_L\alpha_R = \beta'_L \cup \alpha_R$, where $\beta'_L = \beta_L
- \vec{c}$ and $\beta'_R = \beta_R - \vec{c}$.
\end{lemma}

\begin{figure}[htp]
\centering
 \includegraphics[width=0.6\textwidth]{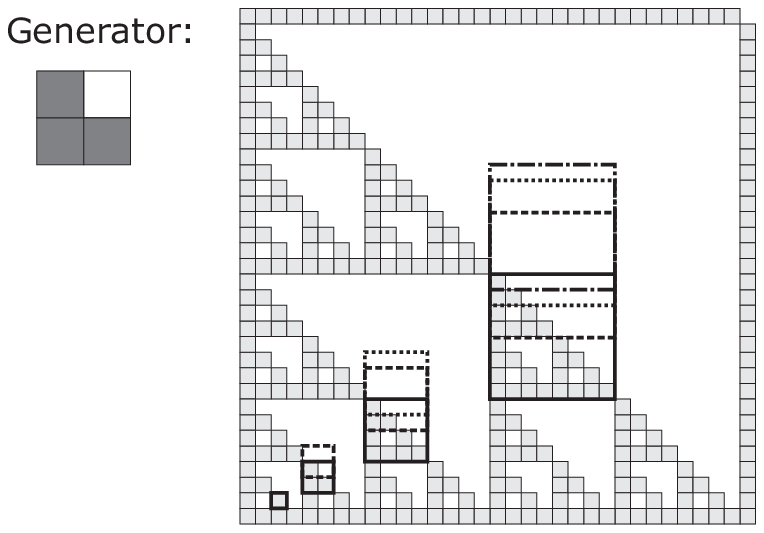}
\caption{In each stage of the Sierpinski triangle, it is possible to define a sequence of closed-rectangular window movies, with the following properties: the number of window movies in the sequence is proportional to the stage number and the set of points contained in each window is unique. }
\label{fig:sierpinski_standard_wml}
\end{figure}

Basically, the reason we do not use the standard WML to prove
Theorem~\ref{thm:main} is because we simply are not able to devise a
unified strategy for finding two closed-rectangular window movies in a
pier-fractal-shaped assembly that (1) have equivalent (up to
translation) bond-forming submovies and (2) contain different
sub-shapes of the assembly. On the one hand, it is trivial to find two
such closed-rectangular window movies in a pier-fractal-shaped
assembly whose sub-shapes are equal. But this does not help us derive
the contradiction that we need to prove Theorem~\ref{thm:main}. On the
other hand, it is also trivial to find two closed-rectangular window
movies that contain different sub-shapes of the assembly, but, as a
result of the self-similarity of pier fractals, do not have equivalent
(up to translation) bond-forming submovies, at which point the
conditions of the hypothesis of the standard WML are no longer
satisfied.

\begin{figure}[htp]
\centering
 \includegraphics[width=4.5in]{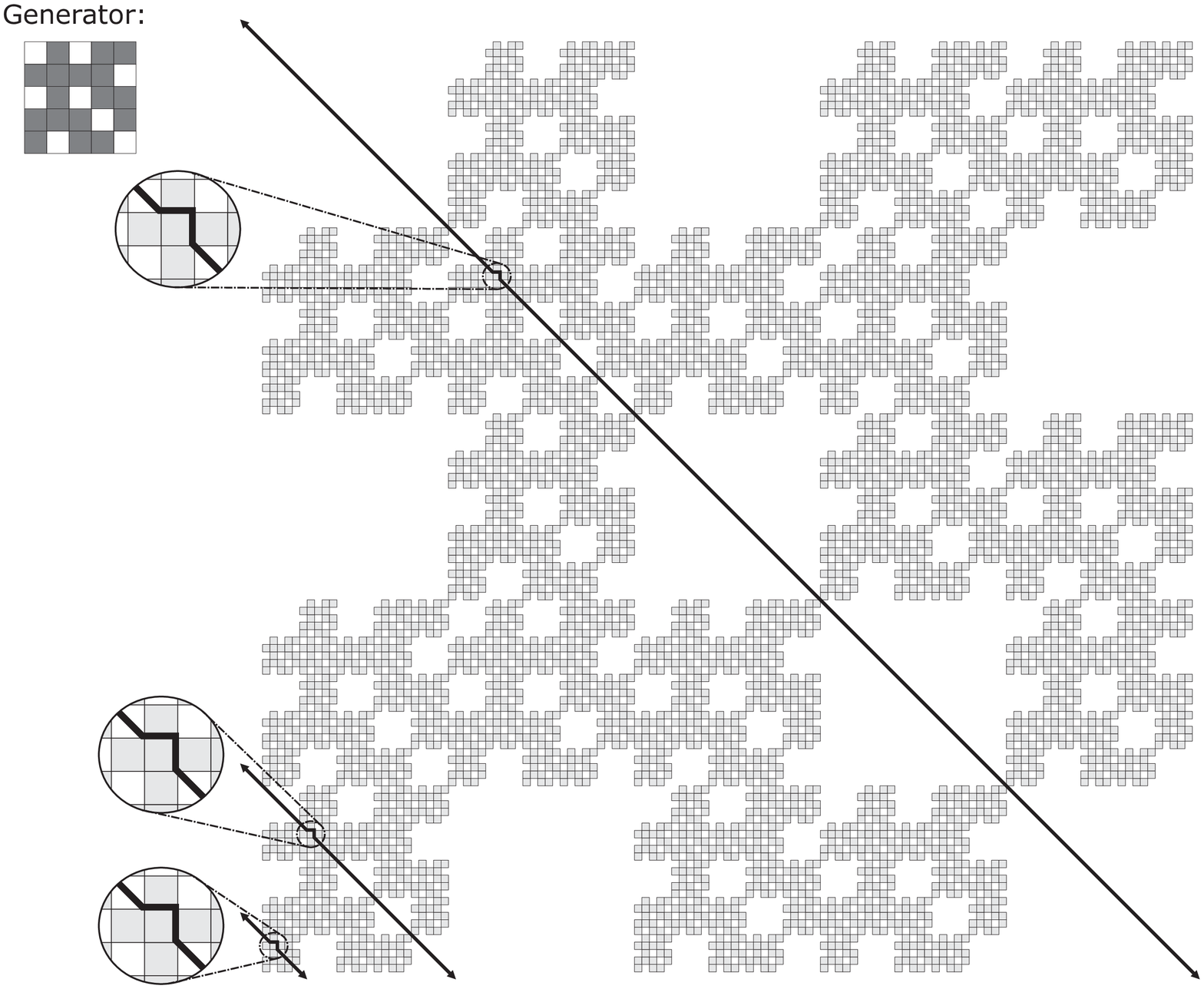}
\caption{A generator for a pier fractal and the first three stages of
  an unscaled version of it. Note that it is possible to apply the
  standard WML to this pier fractal using infinite-open windows. }
\label{fig:interesting_open_window}
\end{figure}

In our attempts to resolve this dilemma, we investigated the use of an
infinite-open window strategy, as opposed to a closed-rectangular
window strategy. But this approach has its own set of technical
challenges.
Fortunately, these challenges can be dismissed! One must simply
observe that, in order to prove Theorem~\ref{thm:main}, one does not
need the ``two-way-assembly-replacement'' power offered by the
conclusion of the standard WML. In fact, in order to derive a
contradiction to prove Theorem~\ref{thm:main}, one merely needs to be
able to replace one portion of a tile assembly with another portion in
a strictly ``one-way'' fashion, i.e., the part of the tile assembly
being used to replace another part does not need to be able to be
replaced by the part of the tile assembly it is replacing. Thus, we
weaken the conclusion and strengthen the hypothesis of the standard
WML to get the Closed WML, which turns out to be much more
accommodating to a unified, closed-rectangular window proof technique
for pier fractals.

It appears that the hypothesis of the standard WML, unlike that of the
Closed WML, is too strong to be able to ``handle'' all pier fractals
under a unified closed-rectangular window proof technique. However, it
is worthy of note that in some special cases, it is possible to use
the standard WML to prove that certain pier fractals do not strictly
self-assemble. For example, it is possible to prove that the
Sierpinski triangle does not strictly self-assemble at any positive
scale factor (see Figure~\ref{fig:sierpinski_standard_wml} for the
proof idea). Next, consider the tree fractal defined by the generator
given in Figure~\ref{fig:interesting_open_window}. In this case, it is
possible to apply the standard WML using an open-infinite window proof
technique (informally depicted in
Figure~\ref{fig:interesting_open_window}). Unfortunately, depending on
the geometry of the particular fractal, neither of the previous two
applications of the standard WML, either with closed-rectangular or
open-infinite windows, immediately generalizes to even the set of all
tree fractals, which is a strict sub-class of pier fractals. Even more
troubling, we suspect that, for the pier fractal whose generator is
shown in Figure~\ref{fig:probably_impossible_generator}, it is not
possible to apply the standard WML, with windows of any shape, to prove that it
does not strictly self-assemble at any positive scale factor.

\begin{figure}[htp]
\centering
 \includegraphics[width=0.75in]{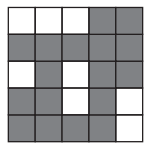}
\caption{How can one apply the standard WML to prove that any scaled
  version of the pier fractal with this generator does not
  strictly self-assemble? }
\label{fig:probably_impossible_generator}
\end{figure}

\section{Conclusion}\label{sec:conclusion}

In this paper, we made three contributions. First, we gave a new
characterization of tree fractals in terms of simple geometric
properties of their generator (see
Section~\ref{sec:appendix}). Second, we proved a new variant of the
Window Movie Lemma in \cite{WindowMovieLemma}, which we call the
``Closed Window Movie Lemma'' (see
Section~\ref{sec:window-movie-lemma}). Third, we proved that no
scaled-up version of any discrete self-similar pier fractal strictly
self-assembles in the aTAM (see Section~\ref{sec:main-result}).

As we pointed out in Section~\ref{sec:generalizations}, the scope of
applicability of the Closed Window Movie Lemma is much wider than the
class of pier fractals. Recall that Corollary~\ref{cor:same_column}
applies the Closed Window Movie Lemma to discrete self-similar
fractals with no pier-like sub-configurations and an arbitrary number
of vertical and horizontal bridges. In future work, we would like to
provide a characterization of the class of all fractals to which the
Closed Window Movie Lemma applies, that is, a strict
super-class of the class of pier fractals. In addition, it would be
satisfying to find a crisp characterization of the differences (if
any) between the scope of applicability of the standard WML and that
of the Closed WML. For instance, we would like to prove our conjecture
that it is not possible to use the standard WML to prove that any
scaled version of the pier fractal whose generator is shown in
Figure~\ref{fig:probably_impossible_generator} does not strictly
self-assemble in the aTAM.

\vspace*{-2mm}\section*{Acknowledgments}
We would like to thank Kimberly Barth and Paul Totzke for contributing to
an earlier proof of Corollary~\ref{cor:tree} that appeared in
\cite{ScaledTreeFractals} and that we generalized to obtain the proof of
our main result in this work.\vspace*{-2mm}

\bibliographystyle{amsplain}
\bibliography{tam}

\ifabstract
\newpage
\section{Appendix}\label{sec:appendix}

\magicappendix
\fi

\end{document}